\documentclass[11pt,letterpaper]{article}

\usepackage{fullpage}
\usepackage{amsmath, amsfonts, amssymb, epsf, epsfig, amsthm,dsfont}
\usepackage{graphicx}
\usepackage{wrapfig}
\usepackage{xspace}
\usepackage{url}
\usepackage[textsize=tiny,disable]{todonotes}
\usepackage[xcolor=pdftex,ulem=normalem]{changes}
\definechangesauthor[Philipp Woelfel]{PW}{blue}
\setauthormarkup[right]{$^{\text{\tiny #1}}$}

\newcommand{\paren}[1]{{\left({#1}\right)}}
\newcommand{\bparen}[1]{{\bigl({#1}\bigr)}}
\newcommand{\Bparen}[1]{{\Bigl({#1}\Bigr)}}
\newcommand{\bbparen}[1]{{\biggl({#1}\biggr)}}
\newcommand{\BBparen}[1]{{\Biggl({#1}\Biggr)}}
\newcommand{\nparen}[1]{({#1})}

\makeatletter
\providecommand{\@prob}{{\mathrm{Prob}}}
\newcommand{\prob}{\@prob}
\newcommand{\Prob}[2][]{\@prob_{#1}\paren{#2}}
\newcommand{\bProb}[2][]{\@prob_{#1}\bparen{#2}}
\newcommand{\BProb}[2][]{\@prob_{#1}\Bparen{#2}}
\newcommand{\bbProb}[2][]{\@prob_{#1}\bbparen{#2}}
\newcommand{\BBProb}[2][]{\@prob_{#1}\BBparen{#2}}
\newcommand{\nProb}[2][]{\@prob_{#1}\nparen{#2}}

\providecommand{\@Exp}{{\mathrm{E}}}
\newcommand{\Exp}[2][]{\@Exp_{#1}\left[{#2}\right]}
\newcommand{\bExp}[2][]{\@Exp_{#1}\bigl[{#2}\bigr]}
\newcommand{\BExp}[2][]{\@Exp_{#1}\Bigl[{#2}\Bigr]}
\newcommand{\bbExp}[2][]{\@Exp_{#1}\biggl[{#2}\biggr]}
\newcommand{\BBExp}[2][]{\@Exp_{#1}\Biggl[{#2}\Biggr]}
\newcommand{\nExp}[2][]{\@Exp_{#1}[{#2}]}

\newcommand{\CondExp}[3][]{\Exp[#1]{\,#2\,\left\vert\vphantom{#2#3}\right.\,#3\,}}

\makeatother

\usepackage{tikz}
\usetikzlibrary{automata,%
positioning,%
calc,%
shapes,%
backgrounds%
}

\usepackage[
  boxruled,
  linesnumbered
]{algorithm2e}
\DontPrintSemicolon
\SetAlgoSkip{normalskip}
\SetAlgoInsideSkip{medskip}

\newcommand{\IlIf}[2]{\KwSty{if} #1 \KwSty{then} #2}
\newcommand{\IlRepeat}[2]{\KwSty{repeat} #2 \KwSty{until} #1}

\def\read{\text{\tt read}\xspace}
\newcommand{\xwrite}{\text{\tt write}\xspace}
\newcommand{\fetchInc}{\text{\tt fetch\&inc}\xspace}
\newcommand{\fetchDec}{\text{\tt fetch\&dec}\xspace}
\newcommand{\fetchSet}{\text{\tt fetch\&set}\xspace}
\newcommand{\llsc}{\text{\sf load-linked/store-conditional}\xspace}
\newcommand{\cas}{\text{\tt compare\&swap}\xspace}
\newcommand{\SC}{\text{\tt SC}\xspace}
\newcommand{\LL}{\text{\tt LL}\xspace}
\newcommand{\collect}{\text{\tt collect}\xspace}
\SetKw{True}{True}
\SetKw{False}{False}

\SetKwFunction{CAS}{CAS}
\SetKwFunction{Scan}{scan}
\SetKwFunction{Update}{update}
\SetKwFunction{Read}{Read}
\SetKwFunction{Write}{Write}
\SetKwFunction{MRread}{MRSWread}
\SetKwFunction{MRwrite}{MRSWwrite}
\SetKwFunction{FAI}{fetch\&inc}
\SetKwFunction{FAD}{fetch\&dec}
\SetKwFunction{LoadBalance}{loadBalance}
\SetKwFunction{Enqueue}{enq}
\SetKwFunction{Dequeue}{deq}
\SetKwFunction{MRwrite}{MRSWwrite}
\SetKwFunction{MRread}{MRSWread}

\newcommand{\eps}{\varepsilon}

\renewcommand{\O}{\ensuremath{{O}}}

\newtheoremstyle{coolstyle}
    {9pt}% Space above
    {9pt}% Space below
    {\slshape}% Body font
    {}% Indent amount
    {\bfseries}% Theorem head font
    {.}% Punctuation after theorem head
    {.5em} % Space after theorem head
    {}% Theorem head spec (can be left empty, meaning `normal')
\theoremstyle{coolstyle}

\newtheorem{theorem}{Theorem}[section]
\newtheorem{claim}[theorem]{Claim}
\newtheorem{lemma}[theorem]{Lemma}

\newtheorem{definition}[theorem]{Definition}

\newtheorem{observation}[theorem]{Observation}

\newtheorem{remark}[theorem]{Remark}

\renewcommand{\AA}{\mathcal{A}}
\newcommand{\BB}{\mathcal{B}}

\newcommand{\EE}{\mathcal{E}}
\newcommand{\HH}{\mathcal{H}}
\newcommand{\MM}{\mathcal{M}}
\newcommand{\PP}{\mathcal{P}}
\newcommand{\QQ}{\mathcal{Q}}

\renewcommand{\SS}{\mathcal{S}}
\newcommand{\VV}{\mathcal{V}}
\newcommand{\WW}{\mathcal{W}}

\newcommand{\IIN}{\mathds{N}}

\newcommand{\vc}{{\vec{c}}}
\newcommand{\vd}{{\vec{d}}}
\newcommand{\tc}{{\tilde{c}}}

\newcommand{\close}[1]{\ensuremath{\text{close}\left(#1\right)}}

\newcommand{\obj}{\ensuremath{\mathit{O}}}
\newcommand{\op}{\ensuremath{\mathit{op}}}
\newcommand{\inv}[1]{\ensuremath{inv(#1)}}
\newcommand{\rsp}[1]{\ensuremath{rsp(#1)}}

\begin{document}

\title{Linearizable Implementations Do Not Suffice for Randomized Distributed Computation }
\author{Wojciech Golab\thanks{Research conducted mostly during a postdoctoral fellowship at the University of Calgary.}\,\,\thanks{
Research partially supported by the Natural Sciences and Engineering Research Council (NSERC) of Canada.}\\
HP Labs\\
\url{wojciech.golab@hp.com}
\and
Lisa Higham{$^\dag$}\\
University of Calgary\\
\url{higham@ucalgary.ca}
\and Philipp Woelfel{$^\dag$}\\
University of Calgary\\
\url{woelfel@ucalgary.ca}
}

\begin{titlepage}
\maketitle
\thispagestyle{empty}
\todo{
   A general note: We have too many LaTeX macro definitions. E.g., I believe the reason for having macros such as {\tt \textbackslash obj} is to allow an easy change of the symbol we use for objects.
   This doesn't work, though: It is impossible for all co-authors to keep track of all the macros, so this will always lead to inconsistencies (e.g., right now {\tt\textbackslash obj} is not used consistently), and then the benefit of using the macro in the first place goes away. Moreover, sometimes we have to use different symbols for objects, and I'd much prefer to type {\tt O} over typing {\tt\textbackslash obj}
   IMO using macros so extensively does not work well for multi-author documents.
   Of course it's o.k.\ to use macros, when it helps to reduce the typesetting effort significantly (as e.g., our macro for {\tt\textbackslash BB} does). But I would much prefer if we could avoid (and get rid of) all macros that don't fall in this category.
}

\begin{abstract}
Linearizability is the gold standard among algorithm designers for deducing the correctness of a distributed algorithm
using implemented shared objects from the correctness of the corresponding algorithm using atomic versions of the same objects.
We show that linearizability does not suffice for this purpose when processes can exploit randomization,
and we discuss the existence of alternative correctness conditions.
This paper makes the following contributions:
\begin{itemize}

\item
Various examples demonstrate that using well-known linearizable implementations of objects (e.g., snapshots)
in place of atomic objects can change the probability distribution of the outcomes that the adversary is able to generate.
In some cases,
an oblivious adversary can create a probability distribution of outcomes for an algorithm with implemented,
linearizable objects,
that not even a strong adversary can generate for the same algorithm with atomic objects.

\item
A new correctness condition for shared object implementations, called \emph{strong linearizability}, is defined.
We prove that a strong adversary (i.e., one that sees the outcome of each coin flip immediately) gains no
additional power when atomic objects are replaced by strongly linearizable implementations.
In general, no strictly weaker correctness condition suffices to ensure this.
We also show that strong linearizability is a local and composable property.

\item
In contrast to the situation for the strong adversary,
for a natural weaker adversary (one that cannot see a process' coin flip until its next operation on a shared object)
we prove that there is no correspondingly general correctness condition.
Specifically, any linearizable implementation of counters
from atomic registers and load-linked/store-conditional objects,
that satisfies a natural locality property,
necessarily gives the weak adversary more power than it has with atomic counters.

\end{itemize}
\end{abstract}

\end{titlepage}

\pagestyle{plain}

\SetAlCapNameSty{textsc}
\SetAlCapSty{textsc}
\SetFuncSty{textsc}

%%%%% Processing file intro.tex
\section{Introduction}

Linearizability is the gold standard among algorithm designers for deducing the correctness of
a distributed algorithm using implemented shared objects from the correctness of the corresponding
algorithm using atomic%
\footnote{
  In this paper, an \emph{atomic} operation is one that happens instantaneously, i.e., it is indivisible.
  But in the literature, the notion of atomicity is not used consistently.
  E.g., in her textbook \cite{Lynch_DistributedAlgorithms1996}, Lynch defines atomic objects to be linearizable,
        but Anderson and Gouda \cite{journals/ipl/AndersonG88} define atomicity in terms of instantaneous operations.}
versions of the same objects.
We explore this in more detail, showing that linearizability does not suffice for this purpose
when processes can exploit randomization.

In an asynchronous distributed system, processes collaborate
by executing an algorithm that applies operations to a collection of shared objects.
If the operations on these objects are atomic, then the result of the execution
is the same as some sequential execution that could arise from an arbitrary interleaving
of the processes' steps.
Alternatively, some objects could be replaced by a set of software methods for the different operations on those objects.
Processes would then invoke the appropriate method in order to simulate the intended atomic operation.
In this case, there is a finer granularity to the interleaving of process steps.
Consequently, we need to be sure that each possible result (e.g., the algorithm's return value for each process) that can arise from using the software methods
could also have arisen if the operations were atomic.

This requirement is ensured if the methods provided for each object constitute
a \emph{linearizable implementation} \cite{her:lin} of the object.
Linearizability is an especially useful and important correctness condition because it is  a \emph{local} property.
That is, if each object in a collection of objects is replaced by its linearizable implementation,
then the result of any execution that can arise from the concurrent use of the whole collection is one
that could have also happened if the objects were atomic.

Linearizable implementations, however, do not preserve the probability distribution of the possible results
as we transform the atomic system to the implemented one.
An \emph{adversary}, which schedules process steps, can ``stretch out'' a method call that was originally an atomic operation,
and concurrently inspect the outcome of other processes' coin flips.
Based on the outcomes, the scheduler can choose between alternative executions of the ongoing method call.
As we will illustrate through examples,
the consequences of this additional flexibility can be powerful and subtle,
allowing the behaviour of the implemented system to differ dramatically from that
of the atomic system.
In particular, the adversary can manipulate executions so that low-probability worst-case results
in the atomic system become much more probable in the implemented system.

We will see that our ability to curtail an adversary's additional power, which it can gain
when atomic objects are replaced by linearizable implementations,
depends in part upon the original power of the adversary.
Various adversaries have been defined in literature,
differing in their ability to base scheduling decisions on the random choices made by the algorithm
(see \cite{Aspnes2003_DistrComp} for an overview of adversary models).
The main results in this paper concern two adversary models.
Informally,
when a process is scheduled by a \emph{strong} adversary, the process executes only its next atomic operation,
whether on a local or a shared object.  (Coins are local objects.)
When a process is scheduled by a \emph{weak} adversary
it executes up to and including its next step on a shared object.
Thus, a strong adversary can intervene between a coin flip and the next step by the same process,
whereas a weak adversary cannot.
Further discussion of these adversaries, including formal definitions, appears in Section \ref{model.sec}.

\subsection*{Summary of contributions}

\noindent
\textbf{1.} Several examples demonstrate that using linearizable implemented objects in place of atomic objects in randomized algorithms allows the adversary to change the probability distribution of results.
Therefore, in order to safely use implemented objects in randomized algorithms,
it does not suffice to simply claim that these implementations are linearizable.

\noindent
\textbf{2.} A new correctness condition for shared object implementations, called \emph{strong linearizability},
which is strictly stronger than linearizability, is defined.
      We prove that a strong adversary against a randomized algorithm using strongly linearizable objects
      has exactly the same power as a strong adversary against the same algorithm using atomic objects.
      Conversely, if the set of histories that arise from a strong adversary scheduling an algorithm with implemented
      linearizable objects
      is  ``equivalent'' to the set of histories that can arise from some strong adversary scheduling the same algorithm
      with atomic objects, then the former set of histories must be strongly linearizable.
      We also show that several known universal constructions of linearizable objects with common progress properties
      (e.g., wait-freedom) provide strong linearizability.
      Finally, we prove that strong linearizability, like linearizability, is both a local and a composable property.

\noindent
\textbf{3.} In contrast to the situation for strong adversaries,
for weak adversaries strong linearizability has no counterpart.
For example, for some randomized algorithms,
weak adversaries always gain additional power when strong counters
(that support \fetchInc and \fetchDec operations)
are replaced with ``natural'' linearizable implementations based on a set of base objects supporting
reads, writes and \llsc operations.
Consequently, to prevent weak adversaries from gaining additional power, the implementation
of the counter would require additional base object types beyond what is necessary for linearizability.
This result is obtained by a technically involved proof;
it holds even for randomized implementations with fairly weak progress conditions (e.g., lock-freedom).

\medskip
Randomization has become an important technique in the design of distributed algorithms;
it allows us to circumvent some substantial impossibilities and complexity lower bounds of deterministic algorithms.
Our results impact the design of randomized algorithms that use shared objects
not directly supported through atomic primitives in hardware.
First, simulating the required shared objects in software using ``only'' linearizable implementations can break the algorithm.
Second, such algorithms are much easier to fix (using strong linearizability)
if they are designed from the outset to work against strong adversaries,
but not so if they are designed only to work against weak adversaries.
Third, since there are strongly linearizable universal constructions using consensus objects, which can be implemented using \cas,
any system that provides \cas in hardware can implement any object in a strongly linearizable way.

\section{Examples}

We begin with two examples to provide intuition and motivation,
and  delay the model details, which are needed for our technical results, until the next section.
The examples illustrate how an adversary in a randomized algorithm gains additional power when atomic objects
are replaced with implemented ones.

\paragraph{Atomic versus linearizable snapshots.}
An $n$ process snapshot object is a vector $(x_1, \ldots, x_n)$ of length $n$
that supports the atomic operations $\Update_p{}$ and $\Scan_p{}$ by any process $p \in \{1, \ldots , n \}$.
Operation $\Update_p(v)$ writes $v$ to $x_p$ while leaving all $x_i, i\neq p$ unchanged; and
$\Scan_p(v)$ returns the vector of values ($x_1, \ldots, x_n$) to $p$.

Initialize a snapshot object for three processes to $(x_p,x_q,x_r) = (0,0,0)$.
Suppose the processes $p,q$ and $r$ are executing the following code,
and the adversary is trying to minimize the sum of the values returned in $p$'s scan.

\begin{quote}
\medskip
\noindent
$p$:  $\Scan_p()$  \\
$r$:  $\Update_r(2)$;\quad  $\Update_r(0)$\\
$q$:  $\Update_q(6)$;\quad  $c :=$uniform-random$\{-1,1\}$;\quad $\Update_q(8 \cdot c)$
\end{quote}
\medskip

To keep the sum in $p$'s \Scan\ low,
the adversary can schedule either both or neither of $r$'s \Update\ operations before $p$'s \Scan.
If the adversary is weak, the same holds for $q$'s \Update\ operations.
Thus, under the best strategy for a weak adversary, the expected value of the sum in $p$'s \Scan is 0.
If the adversary is strong,
its best strategy is to schedule $p$'s \Scan before $q$'s second \Update if $q$'s coin flip returns 1
and after if it returns~$-1$.
Thus, under the best strategy for a strong adversary, the expected value
of the sum in $p$'s \Scan\ is $(6-8)/2=-1$.

Now suppose instead that \Update\ and \Scan\ are implemented from atomic registers
by the well-known wait-free linearizable algorithm due to Afek, Attiya, Dolev, Gafni, Merritt and Shavit \cite{aadgms:snapshots}.
In this algorithm, the snapshot object is implemented as an array $A[1:n]$ of registers.
Let a \emph{collect} denote a series of $n$ atomic reads, one for each element of $A$, in some fixed order.
To perform a \Scan, each process $p$ repeatedly collects until either two successive collects are
identical (a \emph{successful double collect}),
or $p$ observes that another process, say $r$, has executed at least two \Update\ operations to $A[r]$ during $p$'s \Scan.
In the second case, $p$ returns the last \Scan written (as we explain shortly) by $r$ during an \Update (a \emph{borrowed scan}).
To perform an \Update, each process $r$ must first perform a \Scan\ and then write the result of the \Scan\
together with its \Update\ argument into $A[r]$.
This ensures that if a \Scan\ has enough failed double collects, then a borrowed \Scan is possible.
With this implementation, the adversary can maneuver $p$, $q$ and $r$ as shown in Figure~\ref{fig:snapshot}.

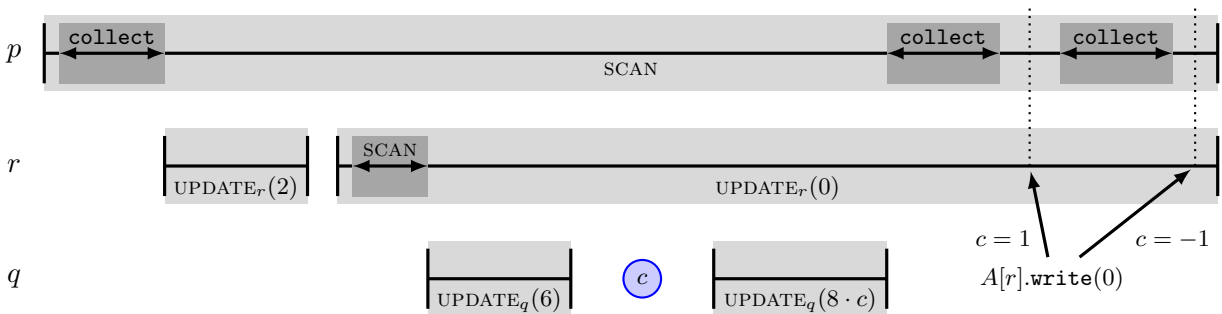
\begin{figure*}[htbp]
\begin{center}
%%%%% Processing file snapshot_example.tex
\begin{tikzpicture}[
>=latex, %use latex arrows
very thick,
on grid,
auto,
]
\footnotesize

\newcommand{\outerop}[4][]{%
  \fill[fill=gray!30] ($(#2)+(0,-.5)$) rectangle ($(#3)+(0,.5)$);
  \draw[-]  (#2) -- node[swap,#1] {#4} (#3);
  \draw[very thick] ($(#2)+(0,-.4)$) -- ($(#2)+(0,.4)$);
  \draw[very thick] ($(#3)+(0,-.4)$) -- ($(#3)+(0,.4)$);
}

\newcommand{\innerop}[4][]{%
  \fill[fill=gray!70] ($(#2)+(0,-.4)$) rectangle ($(#3)+(0,.4)$);
  \draw[<->]  (#2) -- node[#1] {#4} (#3);
}

\node (r) {\normalsize $r$};
\node[below = 1.5 of r] (q) {\normalsize $q$};
\node[above = 1.5 of r] (p) {\normalsize $p$};

\coordinate[right = .4 of p ] (pScanBegin);
\coordinate[right = 16 of p ] (pScanEnd) ;
\outerop{pScanBegin}{pScanEnd}{\Scan}

\coordinate[right =.6 of p ] (pCollect1Begin) ;
\coordinate[right = 1.4 of pCollect1Begin ] (pCollect1End) ;
\innerop{pCollect1Begin}{pCollect1End}{\collect}

\coordinate[right =11.6 of p ] (pCollect2Begin) ;
\coordinate[right = 1.5 of pCollect2Begin ] (pCollect2End) ;
\innerop{pCollect2Begin}{pCollect2End}{\collect}

\coordinate[right =13.9 of p ] (pCollect3Begin) ;
\coordinate[right = 1.5 of pCollect3Begin ] (pCollect3End) ;
\innerop{pCollect3Begin}{pCollect3End}{\collect}

\coordinate[right = 5.5 of q ] (qUpdate2Begin) ;
\coordinate[right = 1.9 of qUpdate2Begin ] (qUpdate2End) ;
\outerop{qUpdate2Begin}{qUpdate2End}{$\Update_q(6)$}

\node[circle,text centered,fill=blue!20,draw=blue!200,thick,
  right = 8.35 of q] (coinflip) {$c$};

\coordinate[right = 9.3 of q] (qUpdate3Begin);
\coordinate[right = 2.3 of qUpdate3Begin ] (qUpdate3End);
\outerop{qUpdate3Begin}{qUpdate3End}{$\Update_q(8\cdot c)$}

\coordinate[right = 2 of r ] (rUpdate1Begin) ;
\coordinate[right = 1.9 of rUpdate1Begin ] (rUpdate1End) ;
\outerop{rUpdate1Begin}{rUpdate1End}{$\Update_r(2)$}

\coordinate[right = 4.3 of r] (rUpdate2Begin) ;
\coordinate[right = 16 of r ] (rUpdate2End) ;
\outerop{rUpdate2Begin}{rUpdate2End}{$\Update_r(0)$}

\coordinate[right = 4.5 of r ] (rScanBegin) ;
\coordinate[right = 1 of rScanBegin ] (rScanEnd) ;
\innerop{rScanBegin}{rScanEnd}{$\Scan$}

\coordinate[right = 13.5 of r] (rWrite1);
\coordinate[right = 15.7 of r] (rWrite2);

\node (label) at ($ (rWrite1)+(.3,-1.5) $) {$A[r].\xwrite(0)$};

\path[->,shorten >= 2pt] (label) edge[pos=0.4] node {$c=1$} (rWrite1);
\path[->,shorten >= 2pt] (label) edge[pos=0.4] node[swap] {$c=-1$} (rWrite2);

\draw[dotted,thick] ($ (rWrite1)+(0,0) $) -- ($ (rWrite1)+(0,2.1) $);
\draw[dotted,thick] ($ (rWrite2)+(0,0) $) -- ($ (rWrite2)+(0,2.1) $);

\end{tikzpicture}
%%%%% Done processing snapshot_example.tex
\end{center}
\vspace{-1em}
\caption{A ``bad'' scheduling using an implemented linearizable snapshot.}
\label{fig:snapshot}
\end{figure*}

In this execution, $r$ applies a \Scan\ that returns a view $S$ with sum~2 as the first part of its second \Update.
Then, the adversary chooses where to schedule the remainder of $r$'s second \Update,
which is the write to $A[r]$ of $(S, 0)$.
If $q$'s coin flip is~$-1$, it schedules this write after $p$'s third collect.
In this case, $p$ will have a successful double collect, which returns a view with sum $2+(-8)=-6$.
If $q$'s coin flip is~1, the adversary schedules $r$'s write between $p$'s second and third collects.
In this case, $p$ will have a failed double collect but will have seen $r$ \Update\ twice.
Accordingly, $p$ borrows $r$'s \Scan, and so $p$'s \Scan\ also returns the view $S$ with sum~2.
Thus, the adversary can force an expected sum in $p$'s \Scan\ of only $(-6+2)/2=-2$.
Notice, furthermore, that only a weak adversary was used to achieve this execution in the system with an implemented
snapshot object.

\paragraph{Atomic versus linearizable registers.}
Since the implemented method calls give the adversary more power than it has when operations are atomic,
we might conjecture that this additional power could be curtailed by appropriately restricting the adversary.
The next example shows that this is not always possible.

Let $R$ denote a multi-valued atomic single-reader/single-writer (SRSW) register initialized to~1.
Let processes $w$ and $p$ execute the following code:

\begin{quote}
\medskip
\noindent
$w$: $R.\Write(2)$;\quad  $c :=$uniform-random$\{0,2\}$;\quad   $R.\Write(c)$  \\
$p$:  $R.\Read()$
\medskip
\end{quote}

Suppose that a strong adversary is trying to minimize the value that $p$ reads.
Then the adversary's best strategy is to have $p$ execute its \Read\ either before or after both of $w$'s \Write\ operations.
In either case, the expected value of $p$'s \Read\ is~1.

Now suppose, instead, that $R$ is implemented
using Vidyasankar's linearizable implementation of
single-reader/single-writer (SRSW) multivalued registers from
SRSW atomic bits \cite{vid:registers}.
In this construction, an array $A[0\dots \ell]$ of
SRSW binary registers is used to represent a register with domain
$\{0, \ldots , \ell\} $.
Value $v$ is represented by $A[v] = 1$ and $A[0]=\dots=A[v-1]=0$.
The implementation is shown in Figure~\ref{fig:multivalued_registers}.

\begin{figure} % using package float
\begin{minipage}[t]{.4\textwidth}
\begin{function}[H]
  $A[v].\xwrite(1)$\;
  \For{$i=v-1,\dots,0$}{
    $A[i].\xwrite(0)$\;
  }
  \caption{Write($v$)}
\end{function}
\end{minipage}\hfill
\begin{minipage}[t]{.5\textwidth}
\begin{function}[H]
  $i:=-1$\;
  \IlRepeat{$A[i].\read()=1$}{$i:=i+1$}\;
  $val:=i$\;
  \For{$i=val-1,\dots,0$}{
    \IlIf{$A[i].\read()=1$}{$val:=i$}
  }
  \Return{$val$}
  \caption{Read()()}
\end{function}
\end{minipage}
\caption{Linearizable implementation of multivalued SRSW registers from atomic bits.}
\label{fig:multivalued_registers}
\end{figure}

Under this implementation, if the register is initialized with the value~1,
the adversary's best strategy is to schedule as follows:
First $p$ reads ``up'' seeing  $A[0] = 0$ and then $A[1] = 1$.
Next, $w$ takes all of its steps, then finally $p$ takes its remaining steps where it reads ``down''.
With probability 1/2, $w$ executed $A[0].\xwrite(1)$ and $p$ will return~0;
with probability 1/2, $w$ executed $A[2].\xwrite(1)$ and $p$ will return~1.
Hence the expected value returned by $p$'s \Read\ is~1/2.

In this example, the adversary makes all its scheduling decisions in advance;
it does not exploit knowledge of the outcome of coin flips while the computation proceeds.
Even reducing the power of the adversary from strong to this weakest \emph{oblivious} one
does not curtail its power sufficiently to retain the expected behaviour of the algorithm
when $R$ is an atomic register.

These examples motivate our central question:
What is required to preserve the behaviour of a randomized algorithm
when atomic operations are replaced by method calls?
The rest of this paper addresses this question.

%%%%% Done processing intro.tex
%%%%% Processing file model.tex
\section{Model and Definitions}
\label{model.sec}

We consider a distributed shared memory system consisting of a set $\PP$ of $n$ processes communicating via a set of globally shared base objects.

A shared object is an instance of a \emph{type}, which supports some set of operations.
Each such operation \op\ consists of an \emph{invocation} including operation arguments, denoted \inv{\op},
and a \emph{matching response} including the return value, denoted \rsp{\op}.
A \emph{type} is defined by a \emph{sequential specification}, which determines the set of sequences of operations that can occur on any object of that type \cite{her:lin}.
A sequence is \emph{valid} for object $\obj$ if it is in the sequential specification of the type of \obj.

In this paper, we restrict ourselves to deterministic types (except for coin objects as described below).
I.e., if $\op_1,\dots,\op_k$ and $\op_1,\dots,\op_{k-1},\op_k'$ are valid sequences and $\inv{\op_k}=\inv{\op_k'}$, then $\rsp{\op_k}=\rsp{\op_k'}$.\todo{Check!}

A process is a sequential thread of control that invokes operations on shared base objects and receives the responses of such operations.
Processes also have access to independent random experiments.
Let $\Omega$ be an arbitrary countable set, called the \emph{coin flip domain}.
A process step can invoke a \emph{flip} operation (with no arguments) on a \emph{coin} object,
which returns a \emph{coin flip} in $\Omega$ as the matching response.

An \emph{implementation} of a \emph{target type} $T$ is a distributed method using other implemented or base objects.
It takes as input the description of an operation invocation, and outputs a response, such that if multiple processes call the method multiple times \emph{sequentially}, then the resulting sequence of method invocations and responses matches the sequential specification of $T$. \todo{Check the  part ``such that if multiple\dots''!}
An implementation is deterministic, if it uses no coin objects; in this paper we consider only deterministic implementations of types.
An \emph{implemented object} is a method that implements a type.\todo{Not sure how to phrase this properly}

Each individual process $p$ executes its program by executing a sequence of operations on shared objects, where the first operation is fixed and the $k$-th operation invocation, $k>1$, is a function of the responses $p$ received from the preceding $k-1$ operations (including flip operations).

Steps of multiple processes interleave, resulting in a \emph{history} $H$, which is a sequence of \emph{steps}, i.e., invocations and responses corresponding to the operations executed by all processes on all base objects and all implemented objects.

Thus, the projection of $H$ onto the steps of any process, $p$, denoted $H|p$, is a sequence of steps consistent with $p$'s program.

We say that an operation $\op$ is \emph{atomic in history $H$}, if $\op$'s invocation is either the last step in $H$, or else is followed immediately in $H$ by a matching response.
(Note that in related literature, an atomic operation is typically represented by a single event.
However, for technical reasons that become more clear in Section~\ref{strongAdversary.sec},  the invocation/response representation is more convenient in this paper.)
Operations on implemented objects are never atomic, while operations on base objects may or may not be atomic.
(We assume that an operation on an implemented object internally applies at least one base object operation.)
A history $H$ is \emph{sequential} if all operations in $H$ are atomic.

A history, $H$, defines a partial \emph{happens before} order $\prec_H$ on its operations, where, for operations $\op$ and $\op'$, $\op\prec_H \op'$ if and only if in $H$ the response of $\op$ occurs before the invocation of $\op'$.
(The relation $\prec_H$ is a total order if and only if $H$ is sequential.)

A sequential history, $H$, is \emph{valid} if, for any object \obj, the projection of $H$ onto the steps associated with \obj, denoted $H|\obj$, is in the sequential specification of the type of \obj.
The new history formed from concatenating history $H$ to the end of history $G$ is denoted $G \circ H$.\todo{Definition of concatenation seems out of place; this paragraph is about sequential histories}

A  history that arises from an algorithm that uses an implemented object \obj\
can be \emph{interpreted} as a history $\Gamma(H)$ of the same algorithm using a base object of the same type:
$\Gamma(H)$ is obtained from $H$ by omitting, for each operation \op\ on \obj, say by process $p$,
  all the steps that appear in $H|p$ after the invocation $\inv{\op}$ and before the matching response $\rsp{op}$.
Thus, for each operation \op\ on \obj\ in $\Gamma(H)$,
\inv{\op} corresponds to the method invocation that simulates operation \op,
\rsp{\op} corresponds to the response of that method call, and
all operations on the base objects within the method call are omitted.
The \emph{set of histories of an implementation} is the set of histories where processes access an object instantiated using the implementation (and no other implemented object).\todo{I find the ``and nothing else'' part rather confusing. It is unclear, whether these histories are interpreted or not; I believe they are not.}
If $\HH$ is a set of histories, then $\Gamma(\HH)=\{\Gamma(H)\,|\,H\in\HH\}$ denotes the set of interpretations of histories in $\HH$.

For correctness, an interpreted history should ``correspond'' to one that could arise from an atomic object.
This is captured by the correctness property called \emph{linearizability} \cite{her:lin}.
(Note that in literature sometimes the term \emph{atomic object} is used to denote a linearizable object, see e.g.\ \cite{Lynch_DistributedAlgorithms1996}.)
An operation, \op, is \emph{complete in a history $H$} if $H$ contains both \inv{\op} and a matching \rsp{\op}.
Since a process is a sequential thread of control,
\todo{``From the definition of a process'' seems wrong. And why do we need this sentence, anyway?}
we see that every operation in $H|p$, except possibly the last one, is complete.
A \emph{linearization} of a history $H$ is a valid sequential history $H'$
that contains all completed operations of $H$ and possibly some non-completed ones (with matching responses added),
and where $\prec_{H'}$ extends $\prec_H$.
A history $H$ is \emph{linearizable} if it has at least one linearization.
(Note that a history containing operations on implemented objects is not linearizable in general because it encodes operations
  on base objects, but its interpretation might be linearizable.)
\todo{Perhaps we should introduce the convention that when we talk about the linearization of a history $H$ we mean the linearization of $\Gamma(H)$?}

An implementation of a shared object type is \emph{linearizable} if its set of histories
contains only histories whose interpretations are linearizable.

Flip operations on a coin object are always atomic,
and return a value from the set $\Omega$ defined earlier.
A vector $\vc=(c_1,c_2,\dots)\in\Omega^\infty$ is called a \emph{coin flip vector}.
History $H$ \emph{observes} the coin flip vector $\vc=(c_1,c_2,\dots)$,
if the $i$-th flip operation in $H$ returns value $c_i$.
For a history $H$ that contains $k$ flip operations,
let $H[k]$ denote the prefix of $H$ that ends with the $k$-th invocation of a flip operation; if fewer than $k$ flips occur during $H$, then $H[k]$ denotes $H$.

The order in which steps of processes interleave is given by a
\emph{schedule}, which is a (possibly infinite) sequence of process IDs.
History $H$ \emph{observes} schedule $\sigma=(\sigma_1,\sigma_2,\dots)$,
if in $H$ the $i$-th step is one executed by process $\sigma_i$.

Schedules are generated by an \emph{adversary}.
Typically, adversaries take the past execution into account to schedule the next process.
We are concerned primarily with two adversaries.
Informally, a \emph{weak adversary} cannot intervene between a flip operation and the next operation invocation by the same process.
This means that in any history, any flip operation by a process $p$ is immediately followed by an invocation step by $p$.
In contrast, a \emph{strong adversary} can use the response of the coin flip to determine which process takes the next step.
The following definitions serve to unify these adversaries, and can easily be seen to capture these informal notions.
An \emph{adversary} is a mapping $\AA:\Omega^\infty\to\PP^\infty$.
An algorithm $\MM$ together with an adversary $\AA$ and a coin flip vector $\vc=(c_1,c_2,\dots)\in\Omega^\infty$
generates the unique history, denoted $H_{\MM,\AA,\vc}$,
that observes the schedule $\AA(\vc)$ and the coin flip vector $\vc$,
and where  all processes perform steps as dictated by $\MM$.

\vspace{-2pt}
\begin{itemize}
\setlength{\itemsep}{-2.8pt}
\item
An adversary without additional restrictions is called an \emph{offline adversary}.
(An offline adversary can ``see'' all the coin flips in advance and can use them to make current scheduling decisions.)
\item
Adversary $\AA$ is \emph{strong for algorithm $\MM$}
if, for any two coin flip vectors $\vc$ and $\vec{d}$ that have a common prefix of length $k$,
 $H_{\MM,\AA,\vc}[k+1]=H_{\MM,\AA,\vd}[k+1]$.
(A strong adversary cannot use future coin flips to make current scheduling decisions.)
\item
Adversary $\AA$ is \emph{weak for algorithm $\MM$}
if it is  strong for algorithm $\MM$ and is additionally constrained so that,
in $H_{\MM,\AA,\vc}$, every flip by process $p$ is followed immediately by the invocation of some operation by $p$.
(A weak adversary cannot use future coin flips or the current coin flip to make the next scheduling decision.)
\item
Adversary $\AA$ is \emph{oblivious} if $\AA$ is a constant function, that is, $\AA(\vc)$ is the same for all $\vc\in\Omega^\infty$.
(An oblivious adversary cannot use coin flips at all to make scheduling decisions.)
\end{itemize}

A strong adversary is commonly considered in the distributed algorithm literature.
Our weak adversary is similar to other adversaries in the literature,
such as that assumed by Chor, Israeli and Li \cite{CIL1987_PODC},
and further discussed by Abrahamson \cite{Abrahamson88_PODC}.
However, while their adversary cannot intervene between flip operations and writes,
it can intervene between flips and reads. (No other atomic operations are considered.)
Our goal is to compare the behaviour of systems with atomic objects to those with implemented objects,
for arbitrary objects that could support stronger operations than just reads and writes.
Consequently, we assume that an adversary treats all operations consistently;
it cannot intervene between a flip and some operations but not others.
Furthermore, always binding a flip operation to the next step of the same process,
instead of binding only if that next step is a write,
serves to strengthen our impossibility result for weak adversaries in Section~\ref{weakAdversary.sec}.

As we compare the powers of different adversaries in the remainder of the paper,
we will refer repeatedly to the following notion of equivalence:
\begin{definition}\label{def_advequiv}
  Let $\MM$ and $\MM'$ be two algorithms and $\AA$ and $\AA'$ be two adversaries.
  We say that $(\MM,\AA)$ and $(\MM',\AA')$ are \emph{equivalent}
  if for any coin flip vector $\vc$,
  there exists a sequential history that is a linearization of $\Gamma(H_{\MM,\AA,\vc})$ and
  of $\Gamma(H_{\MM',\AA',\vc})$.
\end{definition}

Some of the results discussed in Sections~\ref{strongAdversary.sec} and \ref{weakAdversary.sec} refer
to well-known progress requirements.
An implementation of a shared object type is \emph{wait-free} if in any history,
each method call incurs a finite number of steps.
An implementation is \emph{lock-free} if in any history, either
each method call takes finitely many steps, or else infinitely many method calls complete.
An implementation is \emph{terminating} if in any history, either
each method call takes finitely many steps, or else some process that takes finitely many
steps invokes a method call that it does not complete.

%%%%% Done processing model.tex
%%%%% Processing file strongAdversary.tex
\label{strongAdversary.sec}\sloppy

In this section, we discuss a novel technique for limiting the additional power a strong adversary may gain
against an algorithm when atomic objects used by the algorithm are replaced with implemented objects.

\section{Strong Linearizability}
We define a correctness property stronger than linearizability, called \emph{strong linearizability},
and prove that under any strong adversary,
strongly linearizable implementations of shared objects preserve
the probability space of computations of an algorithm using such objects.
We also show that strong linearizability maintains
locality and composability---powerful properties that facilitate algorithm design.

For a set of histories $\HH$, let \close{\HH} denote the prefix-closure of $\HH$.
That is, $G \in \close {\HH} $ if and only if there is a sequence, $S$, of invocation and response steps such that
$G\circ S \in \HH$. (Recall that the operator $\circ$ denotes concatenation.)
Consider a function $f$ that maps a set $\HH$ of histories to a set $\HH'$ of histories.
We say that $f$ is \emph{prefix preserving}, if for any two histories $G,H\in\HH$, where $G$ is a prefix of $H$, $f(G)$ is a prefix of $f(H)$.

\begin{definition}\label{def:prefix-linearizability}
  A set of histories $\HH$ is \emph{strongly linearizable} if there exists a function $f$ mapping
  histories in $\close{\HH}$ to sequential histories, such that\\[-4ex]
  \begin{enumerate}\advance\itemsep-1ex
   \item[(L)] for any $H\in \close{\HH}$, $f(H)$ is a linearization of the interpreted history $\Gamma(H)$, and
   \item[(P)] $f$ is prefix-preserving.
  \end{enumerate}\vskip-1ex
A function satisfying properties (L) and (P) is called a \emph{strong linearization function for $\HH$}.
\end{definition}

An implementation of a type is \emph{strongly linearizable} if the set of histories formed by
interpreting each history in the set of histories of the implementation is strongly linearizable.

We emphasize some differences between the concept of linearizability and strong linearizability:
\begin{enumerate}
 \item In order to determine whether an implementation of a type is linearizable it suffices to look at every single history individually; However, property (P) from the definition of strong linearizability is defined for sets of histories, so we have to consider all possible histories together.
 \item Linearizability is defined in terms of interpreted histories. I.e., it does not matter how an object is implemented, as long as all possible sequences of high-level invocations and responses satisfy the linearizability property.
 For strong linearizability the low-level (i.e., non-interpreted) histories have to satisfy property (P), so the implementation of the object seems to be more important.
\end{enumerate}

In the following we consider sets of histories $\HH$ that are generated by a (fixed) strong adversary $\AA$ for a given algorithm $\MM$.
It will prove helpful to note that in this case it does not make a difference for the strong linearizability of $\HH$ whether $\HH$ is a set of low-level histories or of interpreted histories:
\begin{observation}\label{obs:adversary->interpreted=low-level}
  Let $\MM$ be an algorithm and $\AA$ a strong adversary.
  Then $\HH:=\left\{H_{\MM,\AA,\vc}\,|\,\vc\in\Omega^\infty\right\}$ is strongly linearizable if and only if $\Gamma(\HH)=\left\{\Gamma(H_{\MM,\AA,\vc})\,|\,\vc\in\Omega^\infty\right\}$ is strongly linearizable.
\end{observation}
\begin{proof}
  Assume w.l.o.g.\ that $\HH$ is prefix-closed (and then so is $\Gamma(\HH)$).
  First note that for any $H,H'\in\HH$ the following is true:
  \begin{equation}\label{eq:interpretation-property}
    \text{If $\Gamma(H)=\Gamma(H')$, then either $H$ is a prefix of $H'$ or vice versa.}\tag{$\ast$}
  \end{equation}
  Let $G$ be the longest common prefix of $H$ and of $H'$.
  For the purpose of a contradiction assume that $G$ is a proper prefix of $H$ and of $H'$.
  Since $H$ and $H'$ are generated from the same algorithm $\MM$ and the scheduling of a strong adversary, $G$ must end with the invocation of a flip operation $fl$.
  Hence, the response of that flip operation is in $H$ and $H'$ but its return value is different in these two histories.
  But since all object implementations are deterministic, $fl$ must occur (and respond) in $\Gamma(H)=\Gamma(H')$---a contradiction.

  Now suppose that $\HH$ is strongly linearizable.
  Let $f$ be a strong linearization function for $\HH$.
  For any history $H'\in\Gamma(\HH)$ let $\max_{\HH}(H')$ be the longest history in $\HH$ with $\Gamma(\max_{\HH}(H'))=H'$.
  (By (\ref{eq:interpretation-property}) all histories $H$ with $\Gamma(H)=H'$ are prefixes of $\max_{\HH}(H')$.)
  Define $f'(H')=f(\max_{\HH}(H'))$.
  Then $f'$ is a strong linearization function of $\Gamma(\HH)$:
  If $G'$ is a prefix of $H'$, for $G',H'\in\Gamma(\HH)$, then $\max_{\HH}(G')$ is a prefix of $\max_{\HH}(H')$ and so $f'(G')=f(\max_{\HH}(G'))$ is a prefix of $f'(H')=f(\max_{\HH}(H'))$.
  Moreover, since $f$ satisfies property (L), $f'(H')$ is a linearization of $\Gamma(\max_{\HH}(H'))=H'$. %=\Gamma(H')$.
  Thus, $f'$ satisfies properties (P) and (L).

  Now suppose that $\Gamma(\HH)$ is strongly linearizable and that $g$ is a strong linearization function for it.
  For each history $H\in\HH$ we define $g'(H):=g(\Gamma(H))$.
  Then it is immediate that $g'$ inherits properties (P) and (L) from $g$, so $g'$ is a strong linearization function for $\HH$.
\end{proof}

\subsection{Strong Linearizability is Necessary}
In the following we show that if the power of the set of (at most) strong adversaries is not enhanced by the implemented objects, then the algorithm using implemented objects generates a strongly linearizable set of histories.
\todo{This implies more than what we thought, namely that strong linearizability is necessary even for all weaker (than strong) adversary models. We need to emphasize this somewhere.}
\begin{theorem}\label{thm:prefix-lin-necessary}
Let $\MM$ be an algorithm that uses only atomic objects,
and let $\MM'$ be the algorithm obtained from $\MM$ by replacing some objects with linearizable implementations.
Further, let $\AA'$ be an adversary that is strong for $\MM'$.
If there exists an adversary $\AA$ that is strong for $\MM$ such that $(\AA,\MM)$ and $(\AA',\MM')$ are equivalent,
then $\HH'=\left\{H_{\MM',\AA',\vc}\,|\,\vc\in\Omega^\infty\right\}$ is strongly linearizable.
\end{theorem}

  Due to Observation~\ref{obs:adversary->interpreted=low-level}, it suffices to consider only interpreted histories when proving this theorem.
  Thus, in the following proof we only consider interpreted histories.
  For the ease of notation we simply write $H$ instead of $\Gamma(H)$ for every history $H$ considered.

Let $\HH^\ast=\close{\HH'}$. %\left\{H_{\MM',\AA',\vc}\,|\,\vc\in\Omega^\infty\right\}}$.
Since $(\AA,\MM)$ and $(\AA',\MM')$ are equivalent, and all histories of $\MM$ are sequential,
each history $H'=H_{\MM',\AA',\vc}$, where $\vc\in\Omega^\infty$, has a linearization $\ell(H')=H_{\MM,\AA,\vc}$.
Let $H'\in\HH'$ and let $G'\in\HH^\ast$ be a prefix of $H'$.
Define
$g(G',H')$ to be the shortest prefix of $\ell(H')$ that contains all operations that complete in $G'$.
Two claims help clarify the proof of Theorem~\ref{thm:prefix-lin-necessary}.

\begin{claim}\label{clm:prefix-lin-necessary-lin}
$G=g(G',H')$ is a linearization of $G'$.
\end{claim}
\begin{proof}
Suppose that $\op\prec_{G'} \op'$.
Since $G'$ is a prefix of $H'$, $\op\prec_{H'} \op'$ holds, and thus $\op'\not\prec_{H} \op$ for the linearization $H=\ell(H')$ of $H'$.
History $G$ is a prefix of $H$, and so $\op'\not\prec_{G} \op$.

By construction, $G$ contains all completed operations from $G'$, and so it suffices to show that if $G$ contains an operation, then that operation's invocation occurs also in $G'$.
For contradiction let $\op$ be any operation in $G$ such that $\inv{\op}$ does not occur in $G'$.
By construction, some operation $\op'$ must follow $\op$ in $G$,
and so $\op'$ completes in $G'$ (otherwise, $G$ would not be the shortest prefix of $H$ that contains all operations that complete in $G'$).
However, since $\op\prec_G \op'$ and $G$ is a prefix of $H$, we know that $\op\prec_{H} \op'$ and thus $\op'\not\prec_{H'} \op$.
Hence, in $H'$ $\rsp{\op'}$ occurs only after $\inv{\op}$.
But then, since $G'$ is a prefix of $H'$ that contains $\rsp{\op'}$, $\inv{\op}$ must occur in $G'$ as well---a contradiction.
\end{proof}

\begin{claim}\label{clm:prefix-lin-necessary-extends}
Let $H_c'=H_{\AA',\MM',\vc}, H_d'=H_{\AA',\MM',\vd}\in\HH'$ and $G'\in\HH^\ast$ be a common prefix of both $H_c'$ and $H_d'$.
Then $g(G',H_c')=g(G',H_d')$.
\end{claim}
\begin{proof}
The claim is trivially true if $H_c' = H_d'$, so assume that $H_c' \neq H_d'$.
Suppose that the longest common prefix of $\vc$ and $\vd$ has length $k$.
Since $G'$ is a common prefix of $H_c'$ and $H_d'$, it cannot contain the response of the $(k+1)$-th flip operation. %, and thus it is a prefix of $H_c'[k+1]=H_d'[k+1]$.
Thus, $G'$ is a prefix of both $H_c'[k+1]$ and $H_d'[k+1]$, and furthermore $H_c'[k+1]=H_d'[k+1]$ since $\AA'$ is a strong adversary.
Let $G_c = g(G',H_c')$ and $G_d = g(G',H_d')$.
We will show that neither $G_c$ nor $G_d$ contains the $(k+1)$-th coin flip.
Suppose for contradiction that $G_c$ does.  (The proof for $G_d$ is analogous.)
Since the coin flip is not complete in $G'$, by construction of $G_c$ some operation $\op$ that is complete in $G'$ must follow the coin flip in $G_c$.
Since $G_c$ is a linearization of $G'$ by Claim~\ref{clm:prefix-lin-necessary-lin},
the invocation of the coin flip must precede the response of $\op$ in $G'$.
But that contradicts $G'$ being a prefix of $H_c'[k+1]$, which ends with the coin flip's invocation.
Now since $G_c$ does not contain the $(k+1)$-st coin flip, it is a prefix not only of $\ell(H_c')$
but also of $\ell(H_c'[k+1])$, and similarly $G_d$ is a prefix of $\ell(H_d'[k+1])$.
Since $H_c'[k+1]=H_d'[k+1]$ holds, as noted earlier, this implies that $G_c=G_d$.
\end{proof}

\begin{proof}[Proof of Theorem~\ref{thm:prefix-lin-necessary}]
For all histories $G'\in\HH^\ast$, we define $f(G')=g(G',H')$, where $H'$ is an arbitrary history in $\HH'$ such that $G'$ is a prefix of $H'$.
(By Claim~\ref{clm:prefix-lin-necessary-extends}, all such histories $H'$ yield the same $g(G',H')$.)
We show that $f$ is a linearization function for $\HH'$.
By Claim~\ref{clm:prefix-lin-necessary-lin}, $f$ satisfies property~(L), so it suffices to show that it also satisfies property~(P).
Let $F',G'\in\HH^\ast$, such that $F'$ is a  prefix of $G'$.
Choose an arbitrary history $H'\in\HH'$ such that $G'$ is a prefix of $H'$.
By construction and Claim~\ref{clm:prefix-lin-necessary-extends}, $f(G')$ and $f(F')$ are the shortest prefixes of $\ell(H')$ that contain all completed operations in $G'$ and $F'$, respectively.
Since the set of completed operations in $F'$ is a subset of the completed operations in $G'$, $f(F')$ is a prefix of $f(G')$,
completing the proof of Theorem \ref{thm:prefix-lin-necessary}.
\end{proof}

\subsection{Strong Linearizability is Sufficient for the Strong Adversary}
Under strong linearizability the strong adversary is prevented from using the outcome of the flip
to schedule future events in such a way that they influence the order of past operations in a linearization,
because, once a coin is flipped, the operations that precede the coin flip in the linearization are already determined.
This is made precise in the following theorem, the proof of which appears later in this subsection.

\begin{theorem}\label{thm:prefix-lin->equivalence}
  Let $\MM$ be an algorithm that uses only atomic objects,
  and let $\MM'$ be the algorithm obtained from $\MM$ by replacing some atomic objects with strongly linearizable implementations.
  For any adversary $\AA'$ that is strong for $\MM'$, there exists an adversary $\AA$ that is strong for $\MM$,
  such that $(\MM,\AA)$ and $(\MM',\AA')$ are equivalent.
\end{theorem}
The proof is postponed to a later section.

\subsection{Normalized Strong Linearizations.}
Let $\HH_{\MM, \AA}$ denote the set of all interpreted histories
that are generated by an algorithm $\MM$ and the adversary $\AA$ over all coin flip vectors.
That is
$\HH_{\MM, \AA} = \left\{\Gamma(H_{\MM,\AA,\vc})\,|\,\vc\in\Omega^\infty\right\}$.
A natural way to try to prove Theorem~\ref{thm:prefix-lin->equivalence} would be
to apply the strong linearization function to each history in the set $\HH_{\MM', \AA'}$
to obtain a set $\HH$ of linearizations of $\HH_{\MM',\AA'}$.
Then it would suffice to prove that there is a strong adversary $\AA$ that can generate the histories in $\HH$.

\newcommand{\fl}{\mathit{cf}}
Unfortunately, this is not always possible.
For example, consider an algorithm, $\MM$, where
process $p$ (respectively, $q$) executes the single operation $op_p$ (resp.\ $op_q$)
and process $r$ first executes $op_r$ and then executes a flip operation, $\fl$.
Suppose a strong adversary schedules an implementation of $\MM$
so that, for coin flip $i \in \{0,1\}$, it produces the history:
\begin{eqnarray*}
H'_i &=& \inv{op_p}, \inv{op_q}, \inv{op_r}, \rsp{op_r}, \inv{\fl_r}, \rsp{\fl_r, i}, \rsp{op_p}, \rsp{op_q}
\end{eqnarray*}
Let $G'$ be the common prefix of $H'_0$ and $H'_1$ that ends with \inv{\fl_r}.
Define the function $f$ on $\{G', H'_0, H'_1\}$ by:
\begin{eqnarray*}
f(G') &=& \inv{op_r}, \rsp{op_r}\\
f(H'_0) &=&  \inv{op_r}, \rsp{op_r}, \inv{op_p}, \rsp{op_p}, \inv{op_q}, \rsp{op_q}, \inv{\fl_r}, \rsp{\fl_r, 0}\\
f(H'_1) &=&  \inv{op_r}, \rsp{op_r}, \inv{\fl_r}, \rsp{\fl_r, 1}, \inv{op_q}, \rsp{op_q}, \inv{op_p}, \rsp{op_p}
\end{eqnarray*}
Then, according to Definition~\ref{def:prefix-linearizability},
$f$ can be extended to a strong linearization function for $\{H'_0, H'_1 \}$,
but the histories $f(H'_0)$ and $f(H'_1)$ cannot both be produced by the same strong adversary.
This difficulty is remedied by proving that whenever such a problematic strong linearization function
for a set of histories $\HH$ occurs,
there is another strong linearization for $\HH$ that avoids this problem.
The idea is to move coin flips in $f(H)$ to the earliest point possible,
without violating the happens-before order of $H$.
The following technical lemma makes this precise.

\begin{definition}\label{def:norm-pref-lin}
Let $\HH$ be a strongly linearizable set of histories.
A \emph{normalized strong linearization function for $\HH$} is any
strong linearization function $f^\ast$ for $\HH$
such that:
  \vspace{-3pt}
  \begin{enumerate}
   \item[(N)] for any history $H\in \close{\HH}$, if in $f^\ast(H)$ some flip operation $\fl$ immediately follows
some other operation $\op$, then $\op \prec_H \fl$.
  \end{enumerate}
\end{definition}

\begin{lemma}\label{lem:prefix-linearizable-normal-form}
For any strongly linearizable set $\HH$ of histories there exists a normalized strong linearization function.
\end{lemma}

\begin{proof}
 Let $f$ be a strong linearization function for $\HH$.
  Let $H\in\close{\HH}$ be an arbitrary history and let $cf_1,cf_2,\dots$ be the flip operations in $H$, occurring in that order.
  We obtain $f^\ast(H)$ from $f(H)$ as follows:
  \begin{enumerate}
   \item Let $H'=f(H)$.
   As long as the response of the last operation $op$ in $H'$ does not appear in $H$, remove $op$ from $H'$.
   \item Remove all flip operations from $H'$.
   \item For $i=1,2,\dots$, insert $cf_i$ at the earliest possible position in $H'$, where it doesn't violate $\prec_H$.
  \end{enumerate}

  Let $f'(H)$ be the history obtained from $f(H)$ after execution of Step~1.
  It is immediate from Step~1 that $f'(H)$ satisfies the following:
  \begin{equation}\label{eq:normalization_additional_property}
    \text{The response of the last operation in $f'(H)$ occurs in $H$.}\tag{$\ast$}
  \end{equation}
  Moreover, $f'$ satisfies property (L):
  If $f(H)$ is a linearization of $\Gamma(H)$ that ends with an operation $op$ which does not complete in $H$, then we can remove $op$ without destroying validity or changing the order of any of the remaining pairs of operations in $f(H)$.
  Similarly, property (P) is maintained:
  Since removals from histories occur only at the end, such removals can only violate property (P) for a history $H$ with prefix $G$, if an operation $op$ is removed from $f(H)$ and $op$ occurs also in $f(G)$ but is not removed from $f(G)$.
  If that were to happen, then $f(G)=f(H)$ since $f(G)$ is a prefix of $f(H)$ and $op$ is the last operation of both histories.
  But if $op$ is removed from the end of $f(H)$, then $op$ does not complete during $H$.
  And then $op$ does not complete in $G$, either, and so it is removed from $f(G)$, too.

  Now suppose that $f'$ satisfies (L), (P), and (\ref{eq:normalization_additional_property}) after Step~1.
  We show that after Steps~2 and 3, the resulting mapping, which we call $f^\ast$, satisfies (L), (P), and (N).

  First we show that $f^\ast$ satisfies (N):
  Suppose in Step~3, the flip operation $cf_i$ is inserted immediately after some other operation $op$.
  Then $op\prec_H cf_i$, or else $cf_i$ could have been inserted before $op$.
  For all flip operations $cf_j$ that are inserted after $cf_i$, $cf_i\prec_H cf_j$ holds, and so they are inserted into $H'$ behind $cf_i$.
  Thus, at the end of Step~3, $cf_i$ immediately follows $op$.

  Next we show that $f^\ast$ satisfies (L), i.e., that $f^\ast(H)$ is a linearization of $\Gamma(H)$.
  Since $f'(H)$ and $f(H)$ differ only in the position of the (atomic) flips, and the relative order of flips is preserved, $f(H)$ is valid.
  Now consider two operations $op,op'$ that both appear in $f^\ast(H)$, and where $op\prec_H op'$.
  Then $op$ and $op'$ both occur in $f'(H)$, and since $f'$ satisfies (L), $op\prec_{f'(H)} op'$.
  If neither of the two operations is a flip, then the order of them is preserved when we construct $f^\ast(H)$.
  If one of the operations is a flip, then the insertion rule guarantees that the flip operation is inserted in such a way
  that it does not violate $\prec_H$.
  Finally, only operations that don't complete in $H$ are removed from $f(H)$ in Step~1, so $f'(H)$ and thus also $f^\ast(H)$ contain all operations that complete in $H$.

  It remains to show that $f^\ast$ satisfies (P).
  Let $G$ be an arbitrary prefix of $H$ and let $O$ be the set of non-flip operations in $H$.
  Further, let $O_i=O\cup\{cf_1,\dots,cf_i\}$, for $0\leq i\leq z$, where $z$ is the number of flip operations in $H$.
  We show by induction on $i$, $0\leq i\leq z$, that
  \begin{displaymath}
    f^\ast(G)|{O_i}\text{ is a prefix of }f^\ast(H)|O_i.
  \end{displaymath}
  First consider the base case, $i=0$.
  Note that $O_0=O$.
  By construction, $f^\ast(G)$ contains exactly the same set of operations as $f'(G)$, and $f^\ast(H)$ contains the same set of operations as $f'(H)$.
  Since the order of non-flip operations does not change during Steps~2 and~3, $f^\ast(G)|{O}=f'(G)|O$, and $f^\ast(H)|{O}=f'(H)|O$.
  Since $f'$ satisfies (P), it follows that $f^\ast(G)|O$ is a prefix of $f^\ast(H)|{O}$.

  Now suppose $i\geq 1$ and that $f^\ast(G)|{O_{i-1}}$ is a prefix of $f^\ast(H)|O_{i-1}$.
  Let $G^\ast:=f^\ast(G)|O_i$ and $H^\ast=f^\ast(H)|O_i$.
  Below, we prove the following statement:
 \begin{equation}\label{eq:Proof_of_N}
  \text{for any operation $op$ in $G^\ast$:}\ \
    cf_i\prec_{H^\ast} op
    \ \Leftrightarrow\ %\quad\text{if and only if}\quad
    cf_i\prec_{G^\ast} op.
  \end{equation}
  Thus, if $cf_i$ appears in both $G^\ast$ and $H^\ast$, it appears in the same position.
  On the other hand, if $cf_i$ does not appear in $G^\ast$, then $cf_i\not\prec_{G^\ast} op$ for all operation $op$ that appear in $G^\ast$; hence, in $H^\ast$ $cf_i$ is not inserted before some operation $op$ that appears in $G^\ast$.
  Therefore, from (\ref{eq:Proof_of_N}) we can conclude that $G^\ast$ is a prefix of $H^\ast$, and so it suffices to prove (\ref{eq:Proof_of_N}).

  First suppose $cf_i\prec_{G^\ast} op$, and for the purpose of a contradiction assume
  $cf_i\not\prec_{H^\ast} op$.
  Since all operations in $G^\ast$ occur also in $H^\ast$, $op\prec_{H^\ast} cf_i$.
  Let $op'$ be the operation that immediately precedes $cf_i$ in $H^\ast$.
  Then $op\preceq_{H^\ast} op'$.\footnote{We write $a\preceq b$ to denote ``either $a\prec b$ or $a=b$''.}
  By the semantics of Step~3, $op'\prec_H cf_i$, or else $cf_i$ would have been inserted in front of $op'$.
  Since $cf_i$ completes during $G$, and $G$ is a prefix of $H$, $op'$ completes during $G$ as well.
  Hence, $op'\prec_G cf_i$ and since $G^\ast$ is a linearization of $G|O_i$, $op'\prec_{G^\ast} cf_i$.
  By induction hypothesis, from $op\preceq_{H^\ast} op'$, we have $op\preceq_{G^\ast} op'$ and thus by transitivity $op\prec_{G^\ast} cf_i$---a contradiction.

  Now suppose $cf_i\prec_{H^\ast} op$.
  For the purpose of a contradiction assume that $cf_i\not\prec_{G^\ast} op$.
  Then either $cf_i$ does not occur in $G^\ast$ at all or $op \prec_{G^\ast} cf_i$.
  First consider the latter case.
  Let $op'$ be the last operation in $G^\ast$ that precedes $cf_i$, so $op\preceq_{G^\ast} op'$.
  Since $cf_i$ was not inserted in front of $op'$, by the semantics of Step~3 we have $op'\prec_H cf_i$.
  Since $H^\ast$ is a linearization of $H|O_i$, $op'\prec_{H^\ast} cf_i$.
  By the induction hypothesis, $op\preceq_{H^\ast} op'$, and so by transitivity $op\prec_{H^\ast} cf_i$---a contradiction.

  Finally, assume that $cf_i$ is not in $G^\ast$.
  Let $op'$ be the last operation in $G^\ast$.
  If $op'$ is a flip, then it appears atomic in $G$, and thus completes in $G$.
  Otherwise, $op'$ is also the last operation in $f'(G)$, and thus by (\ref{eq:normalization_additional_property}), $op'$ completes in $G$.
  On the other hand, since $cf_i$ is a atomic but does not appear in $G^\ast$, it does not appear in $G$, either.
  Since $G$ is a prefix of $H$, it follows that $op'\prec_H cf_i$.
  But then, since $H^\ast$ is a linearization of $H|O_i$, $op'\prec_{H^\ast} cf_i$.
  Since $op$ is in $G^\ast$, but $op'$ is the last operation in $G^\ast$, $op\preceq_{G^\ast} op'$.
  By the induction hypothesis, $G^\ast|\{op,op'\}$ is a prefix of $H^\ast|\{op,op'\}$, and so $op\preceq_{H^\ast} op'\prec_{H^\ast} cf_i$---a contradiction.
\end{proof}

\subsection{Strong Linearizability is a Local Property.}
We could now proceed with the proof of Theorem~\ref{thm:prefix-lin->equivalence}.
However, the proof is simplified by exploiting a locality property for strong linearizability.
Herlihy and Wing proved the following locality property for linearizability \cite{her:lin}:
A history $H$
over multiple shared objects is linearizable if, for each such object $\obj$ the history $H|\obj$ is linearizable.
Thus, to establish that any result that can arise from an algorithm using implemented objects,
is a result that could be produced by the same algorithm using atomic objects,
it suffices to show separately for each shared object, \obj,
that the implementation of any algorithm over \obj\ is linearizable.
The analogous property for strong linearizability is
given in the following lemma.
For a shared implemented object $\obj$ and a history $H$, we denote by $H\|\obj$
the history $H$ projected to all steps that any process executes while performing an operation on $\obj$, i.e., all steps executed by the process during an interval that starts with an invocation on $\obj$ and ends with the matching response on $\obj$.

\begin{theorem}\label{thm:locality}
  A set of histories $\HH$ over implemented objects $\obj_1, \obj_2, \ldots, \obj_n$
  is strongly linearizable if for each object $\obj_i$, $1 \leq i \leq n$,
  the set $\HH_i = \left\{(H\|\obj_i)\,\mid\,H \in \HH\right\}$ is strongly linearizable.\todo{Definition of ``history over implemented objects\dots'' is unclear. Here, we want histories that contain all invocations and responses on the implemented objects as well as the base objects used in the implementation.}
\end{theorem}

  \begin{remark}\todo{PLEASE LOOK AT THIS REMARK. This is quite odd.}
     Typically one would apply the theorem for objects $O_1,\dots,O_n$ that are implemented from distinct sets $\BB_1,\dots,\BB_n$ of base objects.
     This is not required for the correctness of the lemma, though.
     However, suppose that $n=2$ and the implementations of $O_1$ and $O_2$ share the same base object $B$.
     It is easy to construct a history $H$ such that $(H\|O_1)|B$ and $(H\|O_2)|B$ are valid, but $H|B$ is not.
     Or, it might be the case that $H|B$ is valid, but $(H\|O_1)|B$ is not.
     However, even if $\HH$ contains such a history, the lemma is correct, because a strong linearization function $f$ maps $H$ to a history $f(H)$ that does not contain any steps on $B$.
  \end{remark}

\begin{proof}[Proof of Theorem~\ref{thm:locality}]
Let $H^{(k)}$ denote the prefix of history $H$ that has length $\min\{k,|H|\}$, and
let $\HH^{(k)}$ be the subset of all histories in $\close{\HH}$ that have length at most $k$.
First, we use the strong linearization functions $f_i$ for each object $\obj_i$ to
inductively construct functions $f^{(k)}$,
each of which maps each history $H \in \close{\HH}$ to a sequential history, and satisfies:
\begin{enumerate}
 \item[(a)] for any object $\obj_i$, $f^{(k)}(H)|O_i = f_i(H^{(k)}\|O_i)$; and
 \item[(b)] $f^{(k)}$ is a strong linearization function for $\HH^{(k)}$; and
 \item[(c)] $f^{(k)}(H)=f^{(k)}(H^{(k)})$; and
 \item[(d)] if $k\geq 1$, then $f^{(k-1)}(H)$ is a prefix of $f^{(k)}(H)$.
\end{enumerate}

For the basis, define $f^{(0)}(H)=\varepsilon$, where $\varepsilon$ denotes the empty history.
Properties (a) through (d) clearly hold for $f^{(0)}$.
Now consider $k\geq 1$ and suppose that $f^{(k-1)}$ satisfies (a)-(d).
Construct $f^{(k)}$ as follows.
For any history $H$, if $|H| < k$ then define $f^{(k)}(H) = f^{(k-1)}(H)$.
Otherwise let $\beta$ be the $k$'th step of $H$, and suppose that $\beta$ is a step in $H\|\obj_j$.
Since $f_j$ is a strong linearization function for $\HH_j$,
$f_j\paren{H^{(k-1)}\|O_j}$ is a prefix of $f_j\paren{H^{(k)}\|O_j}$.
\todo{Below this, I replaced several $|$ with $\|$ without indicating it.}
Consequently, there is a sub-history $\lambda$ (possibly empty) satisfying
\begin{equation}\label{eq:def_lambda}
  f_j\paren{H^{(k)}\|O_j}=f_j\paren{H^{(k-1)}\|O_j}\circ\lambda
\end{equation}
For this case (i.e., $|H| \geq k$) define
\begin{equation}\label{eq:loclity-i.s.}
  f^{(k)}(H)=f^{(k-1)}(H)\circ\lambda
\end{equation}

For all $k \in \IIN \cup \{0\}$,
property (c) is satisfied because $f^{k}(H)$ is uniquely determined by the prefix of length $\min\{k,|H|\}$ of $H$.
Property (d) follows immediately from (\ref{eq:loclity-i.s.}).
We now show that properties (a) and (b) are preserved by the inductive step.

\paragraph{Property~(a):}
First consider an object $O_i$, where $i\neq j$.
Since $\beta$ and all events in $\lambda$ belong \todo{``belong to operations on $O_j$'' means more than just the invocation/response of an operation on $O_j$. This should be made clear somewhere?}
to operations on $O_j$, we get:
\noindent
\begin{displaymath}\begin{split}
 f^{(k)}(H)|O_i
 &\stackrel{(\ref{eq:loclity-i.s.})}{=}
 \Bparen{f^{(k-1)}(H)\circ\lambda}\Big|O_i
 =
 f^{(k-1)}(H)|O_i
 \stackrel{I.H.}{=}
 f_i\paren{H^{(k-1)}\|O_i}
 \\ &=
 f_i\bbparen{\paren{H^{(k-1)}\circ\beta}\|O_i}
 =
 f_i\paren{H^{(k)}|O_i}.
\end{split}\end{displaymath}
Now consider the object $O_j$.
By definition of $\lambda$, the induction hypothesis, and construction of $f^{(k)}$:
\begin{displaymath}\begin{split}
  f_j\paren{H^{(k)}\|O_j}
  \stackrel{(\ref{eq:def_lambda})}{=}
  f_j\paren{H^{(k-1)}\|O_j}\circ\lambda
  \stackrel{I.H.}{=}
  \paren{f^{(k-1)}\paren{H}|O_j}\circ\lambda
  =
  \paren{f^{(k-1)}\paren{H}\circ\lambda}|O_j
  =
  f^{(k)}(H)|O_j.
\end{split}\end{displaymath}

\paragraph{Property~(b):}
We first show that the projection of $f^{(k)}$ on $\HH^{(k)}$ satisfies property (P).
Let $H$ be an arbitrary history of length at most $k$, and let $G$ be a prefix of $H$.
If $|G|=k$, then $G=H$ so trivially $f^{(k)}(G)$ is a prefix of $f^{(k)}(H)$.
Now consider the case $|G|<k$, and thus $G=G^{(k-1)}$.
Then by construction, $f^{(k-1)}(G)=f^{(k)}(G)$, and so
\begin{displaymath}
  f^{(k)}(G)
  =
  f^{(k)}\paren{G^{(k-1)}}
  =
  f^{(k)}(G)
  =
  f^{(k-1)}\paren{G^{(k-1)}},
\end{displaymath}
which by the induction hypothesis is a prefix of
\begin{displaymath}
f^{(k-1)}\paren{H^{(k-1)}}
  \stackrel{(c)}=
  f^{(k-1)}\paren{H}.
\end{displaymath}
By (d), this is a prefix of $f^{(k)}(H)$.

We now prove that $f^{(k)}$ satisfies property (L) for every history $H\in\HH^{(k)}$.
Consider an operation $op$ on some object $O_i$ that completes during $H$ and thus also during $\Gamma(H)$.
Then $op$ completes during $\Gamma(H\|O_i)$, and also during its linearization $f_i(H\|{O_i})$.
Thus, by (a) $op$ completes during $f^{(k)}(H)$.
Now consider two operations $op_1,op_2$ in $\Gamma(H)$ such that $op_1\prec_{H} op_2$.
For the purpose of a contradiction, assume that $op_2$ appears before $op_1$ in $f^{(k)}(H)$.
First consider the case that $op_1$ and $op_2$ both appear in $f^{(k-1)}(H)$.
Since by (c) this history is equal to $f^{(k-1)}(H^{(k-1)})$, we conclude from (b) that $op_1$ precedes $op_2$ in $f^{(k-1)}(H)$.
But then by (d), this must also be the case in $f^{(k)}(H)$---a contradiction.

Now assume that not both of $op_1$ and $op_2$ occur in $f^{(k-1)}(H)$.
Then by (d) and the assumption that $op_2$ precedes $op_1$ in $f^{(k)}(H)$, we know that $op_1$ does not appear in $f^{(k-1)}(H)$.
Hence, according to (\ref{eq:loclity-i.s.}) we can write $f^{(k)}(H)=f^{(k-1)}(H)\circ\lambda$ and $op_1$ appears in $\lambda$.
By the induction hypothesis for (b), $op_1$ does not complete during $H^{(k-1)}$.
From the assumption $op_1\prec_H op_2$, it follows that $op_2$ does not get invoked during $H^{(k-1)}$.
Hence, $op_2$ cannot occur in $f^{(k-1)}(H)=f^{(k-1)}(H^{(k-1)})$ either or else $f^{(k-1)}$ would violate (b).
Therefore, $op_1$ and $op_2$ must both have events in $\lambda$.
By construction, $\lambda$ is a sub-history of $\Gamma(H\|O_i)$ for one object $O_i$.
Hence, $op_1$ and $op_2$ are both operations on the same object $O_i$ and so $op_1\prec_{H|O_i} op_2$.
On the other hand, by the assumption that $op_2$ precedes $op_1$ in $f^{(k)}(H)$ and since by (a) $f^{(k)}(H)|O_i=f_i(H^{(k)}\|O_i)$ and since $f_i$ satisfies (P), $op_2$ occurs before $op_1$ in $f_i(H\|O_i)$.
But then $f_i(H\|O_i)$ is not a linearization of $H\|O_i$, contradicting the assumption that $f_i$ is a strong linearization function for $\HH_i$.

So we have proved that the functions $f^{(k)}$ satisfy properties (a) through (d).

The final step of the proof is to use the functions $f^{(k)}$ to define the function $f$ so that it is
a strong linearization function for $\HH$ (which may contain infinite histories).
By property (d), for $k\geq 0$, there is a sequence $\zeta^{(k)}(H)$ satisfying
$f^{(k)}(H)\circ\zeta^{(k)}(H)=f^{(k+1)}(H)$.
We use this sequence to define $f$ as follows:
\begin{displaymath}
  f(H)=
  \begin{cases}
    f^{(|H|)}(H) &\text{if $|H|$ is finite, and}\\
    f^{(0)}(H)\circ\zeta^{(0)}(H)\circ \zeta^{(1)}(H)\circ \zeta^{(2)}(H)\circ\cdots & \text{if $|H|=\infty$.}
  \end{cases}
\end{displaymath}
(Note that if $|H|$ is finite, then $f(H)=f^{(0)}(H)\circ \zeta^{(1)}(H)\circ\dots\circ\zeta^{(|H|)}(H))$.)

It now remains to confirm that $f$
satisfies properties (L) and (P) of Definition~\ref{def:prefix-linearizability}.

If $G$ is a prefix of length $\ell$ of some history $H\in\HH$, then
\begin{displaymath}
  f(G)
  =
  f^{(\ell)}(G)
  \stackrel{(c)}{=}
  f^{(\ell)}(H)
  =
  f^{(0)}(H)\circ\zeta^{(0)}(H)\circ\cdots \circ\zeta^{(\ell)}(H)
\end{displaymath}
is a prefix of $f(H)$.
Hence, $f$ satisfies property (P).

To show that $f$ satisfies property (L) consider an arbitrary history $H\in\HH$.
First note that every operation $op$ that completes in $\Gamma(H)$ completes in $f(H)$, too:
Let $j$ be the position of the response of $op$ in $H$.
Then by (b) and (c), $f^{(j)}(H)=f^{(j)}(H^{(j)})$ is a linearization of $H^{(j)}$.
By (d), $f^{(j)}(H)$ is a prefix of $f(H)$ that contains $op$, since $op$ completes in $\Gamma(H^{(j)})$.

Now suppose that $f(H)$ is not a linearization of $\Gamma(H)$.
Then there is a finite prefix $G'$ of $f(H)$, such that either $\prec_{G'}$ is not compatible with $\prec_H$, or $G'$ is not valid.
Let $k$ be an arbitrary large enough integer such that $G'$ is a prefix of $f^{(k)}(H)$.
By (c), $f^{(k)}(H)=f^{(k)}(H^{(k)})$, and so by (b), $f^{(k)}(H)$ extends $\prec_{H^{(k)}}$ and is valid.
Clearly, the same is true for every prefix of $f(H^{(k)})$.
Hence, $G'$ is valid and extends $\prec_{H^{(k)}}$,
and so $G'$ also is compatible with $\prec_H$.
We conclude that $f(H)$ is a linearization of $H$.
\end{proof}

\subsection{Proof that Strong Linearizability is Sufficient for the Strong Adversary.}
We now have the tools to prove the core theorem concerning correctness under the strong adversary.

\begin{proof}[Proof of Theorem~\ref{thm:prefix-lin->equivalence}]
  Let $\HH'=\left\{H_{\MM',\AA',\vc}\,|\,\vc\in\Omega^\infty\right\}$.
  By Theorem~\ref{thm:locality}, since all object implementations are strongly linearizable, $\HH'$ is strongly linearizable.
  By Lemma~\ref{lem:prefix-linearizable-normal-form},
  there is a normalized strong linearization function, say $f$, for $\HH'$.
  For each coin flip vector $\vc$, let $H'_\vc=H_{\MM',\AA',\vc}$ and $H_\vc=f(H'_\vc)$.
  The set of all pairs $(\vc, H_\vc)$ with $\vc\in\Omega^\infty$ defines an adversary $\AA$, where $\AA(\vc)$ is the schedule observed by $H_\vc$.
  Then $H_\vc=H_{\MM,\AA,\vc}$ for all $\vc\in\Omega^\infty$, and $H_\vc$ and $\Gamma(H'_\vc)$ have a common linearization (namely $H_\vc$).
  Hence, $(\MM, \AA)$ and $(\MM',\AA')$ are equivalent.

  It remains to prove that $\AA$ is strong for $\MM$.
  That is, for any two coin flip vectors, $\vc=(c_1,c_2,\dots)$ and $\vec{d}=(d_1,d_2,\dots)$,
  $c_i=d_i$ for $0\leq i\leq k$ implies $H_\vc[k+1]=H_\vd[k+1]$.

  Since $\AA'$ is a strong adversary, it follows that $H'_\vc[k+1]=H'_\vd[k+1]$.
  If $H'_{\vc}$ contains fewer than $k+1$ flip operations, then $H'_\vc=H'_\vc[k+1]=H'_\vd[k+1]=H'_\vd$, and so the claim holds.

  Now suppose that $H'_{\vc}$ contains at least $k+1$ flip operations.
  Let $\fl$ be the $(k+1)$-st flip in $H'_\vc$ and \inv{fl}  its invocation.
  Since $H_{\vc}$ is a linearization of $\Gamma(H'_{\vc})$, and since coin flips are atomic,
  the $(k+1)$-st flip operation in $H_{\vc}$ is also $\fl$, and \inv{fl} is the last step in $H_\vc[k+1]$.

Let $G'$ and $G$ be the prefixes of $H'_\vc[k+1]$ and $H_\vc[k+1]$, respectively, each of which ends with the last step preceding \inv{fl}.
Since $G'$ is a prefix of $H'_\vc$, property~(P) for $f$ (see Definition~\ref{def:prefix-linearizability})
ensures that $f(G')$ is a prefix of $f(H'_\vc)=H_\vc$.
Thus $G$ and $f(G')$ are both prefixes of $H_\vc$.
We will now show that the last operation of $G$ is in $f(G')$ and the last operation of $f(G')$ is in $G$,
from which it follows that $f(G') = G$.

  Let $\op$ be the last operation in $G$.
  Then in $H_\vc$, operation $op$ immediately precedes the flip operation $\fl$.
  Consequently, property~(N) (see Definition~\ref{def:norm-pref-lin}) ensures that $op\prec_{H'_\vc} \fl$,
  and so $\op$ completes during $G'$.
  Since $f(G')$ is a linearization of $G'$, $\op$ appears in $f(G')$ as well.

  Now consider the last operation $\op'$ in $f(G')$.
  Since $f$ has property~(P)
  and coin flips appear atomic,
  $f(G')$ is a prefix of $H = f(G'\circ \inv{\fl} \circ \rsp{\fl} )$.
  There can be no operation $\op''$ between $\op'$ and $\fl$ in $H$, because if there were, $\op''$ would not complete in  $f(G')$
  and thus could not satisfy $\op''\prec_{H'_\vc} \fl$, contradicting property~(N) for $f$.
  Thus, $\op'$ immediately precedes $\fl$ in $H$, and hence also in $f(H'_\vc)$.
  Then by Definition~\ref{def:norm-pref-lin}, $\op'\prec_{H'_\vc} \fl$.
  Therefore, $\op'$ precedes $\fl$ also in $f(H'_\vc)=H_\vc$ and so $\op'$ appears in $G$.

Thus, we conclude that $f(G')=G$, and hence $f(G') \circ \inv{fl} = G \circ \inv{fl}  = H_\vc[k+1]$.
By symmetry we also have $ f(G') \circ \inv{fl}  = H_\vd[k+1] $, which completes the proof.
\end{proof}

\subsection{Timed Executions, Linearization Points, and Composability}\label{sec:timed}
\todo{This Section is completely new}
A \emph{timed execution} is a (possibly infinite) sequence of pairs $E=\bigl((s_1,t_1),(s_2,t_2),\dots\bigr)$, where each $s_i$ is an invocation or response step and each $t_i$ is a number in $\mathds{R}$ satisfying
\begin{itemize}
\item $t_1\leq t_2\leq\dots$, and
\item if $t_i=t_{i+1}$, then $s_{i+1}$ is the response of an (atomic) operation with invocation $s_i$.
\end{itemize}
The timed execution $E$ corresponds to a history $H(E)=(s_1,s_2,\dots)$ together with a \emph{timing function} $t_E:\{s_1,s_2,\dots\}\to\mathds{R}$, $t_E(s_i)=t_i$.
We say that step $s_i$ occurs at \emph{time $t(s_1)$} in execution $E$.
If an operation $op$ is atomic in $H(E)$, then we say that $op$ occurs atomically at time $t_E\bigl(inv(op)\bigr)$.
Every operation $op$ of the timed execution $E$ can be associated with an interval $I_E(op)\subseteq\mathds{R}$, where $I_E(op)=\left[t_E\bigl(inv(op)\bigr),t_E\bigl(rsp(op)\bigr)\right]$, if $rsp(op)$ occurs in $E$, and $I_E(op)=\left[t_E\bigl(inv(op)\bigr),\infty\right]$, otherwise.

For a timed execution $E=\bigl((s_i,t_i)\bigr)_{i=1,2,\dots}$ let $\Gamma(E)$ denote the interpretation of $E$, i.e., $\Gamma(E)$ is the timed execution that contains only the pairs $(s_j,t_j)$, where $s_j\in\Gamma\bigl(H(E)\bigr)$.
Define $\Phi(H)$ to be the set of operations whose invocations occur in history $H$, and for a timed execution $E$ let $\Phi(E)=\Phi\bigl(H(E)\bigr)$.

Now consider an arbitrary timed execution $E$ and the corresponding history $H=H(E)$.
It is well-known that a valid sequential history $L$ is a linearization of $\Gamma(H)$ if and only if
\todo{Check this}
every operation $op\in \Phi\bigl(\Gamma(E)\bigr)$ can be mapped to a point $pt(op)\in\mathds{R}\cup\{\infty\}$, such that for any $op\in \Phi\bigl(\Gamma(E)\bigr)$ and any distinct $op_1,op_2\in \Phi(L)$:
\begin{enumerate}
\item[(a)] $pt(op)\in I_{E}(op)$, and
\item[(b)] if $op_1\preceq_L op_2$, then $pt(op_1)< pt(op_2)$.
\end{enumerate}

If $pt:\Phi\bigl(\Gamma(E)\bigr)\to\mathds{R}\cup\{\infty\}$ is a mapping satisfying both, (a) and (b), then we say that $pt$ maps the operations $op\in\Phi\bigl(\Gamma(E)\bigr)$ to their \emph{linearization points}.
Note that by (a) all linearization points $pt(op)$, with $pt(op)<\infty$, are distinct.
If $pt$ maps the operations in $\Gamma(E)$ to their linearization points, then we denote by $L(E,pt)$ the timed execution, where every operation $op\in\Phi(E)$, $op<\infty$, occurs atomically at time $pt(op)$.

Two timed executions $E,E'$ are \emph{isomorphic}, if $H(E)=H(E')$.
For a set $\EE$ of executions, let $\close{\EE}$ the set of all executions $E'$ that are isomorphic to a prefix in $\EE$.

\begin{definition}\label{def:slm}
Let $\EE$ be a set of timed executions and for every execution $E\in\close{\EE}$ let a mapping $pt_E:\Phi\bigl(\Gamma(E)\bigr)\rightarrow\mathds{R}\cup\{\infty\}$ be given.
We say that the mappings $pt_E$, $E\in\close{\EE}$, map the operations of $E$ to their \emph{strong linearization points}, if for all $D,E\in\close{\EE}$
  \begin{enumerate}
  \item[$(L')$] $pt_E$ maps the operations in $\Gamma(E)$ to their linearization points, and
  \item[$(P')$]
    if $D$ is a prefix of $E$ then $L(D,pt_D)$ is a prefix of $L(E,pt_E)$.
  \end{enumerate}
  In this case, $\EE$ is called \emph{strongly linearizable}.
\end{definition}

It is immediate that we get the following characterization of strong linearizability of sets of histories:
\begin{lemma}\label{lem:strong-linearizability-linearization-points}
  Let $\EE$ be a set of timed executions and $\HH=\bigl\{H(E)\,|\,E\in\EE\bigr\}$.
  Then $\HH$ is strongly linearizable if and only if $\EE$ is strongly linearizable.
\end{lemma}
\begin{proof}
  Assume w.l.o.g.\ that $\EE=\close{\EE}$; thus $\HH$ is prefix-closed.
  First suppose that $\EE$ is strongly linearizable, i.e., there exist mappings $pt_E$, $E\in\EE$, that map the operations in $E$ to their strong linearization points.
  For each history $H\in\HH$, let $E_H$ be the timed execution with $H(E_H)=H$, where the $i$-th step of $H$ occurs at time $i$.
  (Such an execution $E_H$ exists because $\EE=\close{\EE}$ is by definition closed under isomorphism.)
  Then $E_H\in\EE$ and we can define $f(H)=H\left(L\bigl(E_H,pt_{E_H}\bigr)\right)$.
  By $(L')$, $f(H)$ is a linearization of $\Gamma\bigl(H(E_H)\bigr)=\Gamma(H)$.
  Moreover, if $G$ is a prefix of $H$, then by construction $E_G$ is a prefix of $E_H$, and so from $(P')$ we immediately get that $f(G)$ is a prefix of $f(H)$.
  Hence, $f$ is a strong linearization function for $\HH$.

  Now suppose that $\HH$ is strongly linearizable and let $f$ be a strong linearization function of $\HH$.
  Consider an arbitrary timed execution $E\in\EE$, let $H=H(E)$ be the corresponding history, and let $t_E$ be the corresponding timing function.
  (I.e., if $E=\bigl((s_1,t_1),(s_2,t_2),\dots\bigr)$, then $t_E(s_i)=t_i$.)
  Further, let $k=|f(H)|\in\mathds{N}\cup\{0,\infty\}$, and let $op_i$, $1\leq i\leq k$, be the $i$-th operation in $f(H)$.
  Finally, for any point time $t\in\mathds{R}$ let $T^\ast(t)=(t+t')/2$, where $t'>t$ is the point in time of the first step that occurs in $E$ after time $t$, and $T^\ast(t)=t+1$ if no step occurs in $E$ after point $t$.

  We inductively define
  \begin{align*}
    & pt_{E}(op_1)=t_E\bigl(inv(op_1)\bigr),\quad\text{and}\\
    & pt_{E}(op_{i})=\max\left\{t_E\bigl(inv(op_i)\bigr),\,T^\ast\bigl(pt_{E}(op_{i-1})\bigr)\right\}\quad\text{for $1<i\leq k$.}
  \end{align*}
  For every operation $op\in \Phi\bigl(\Gamma(E)\bigr)-\Phi(f(H))$ we define $pt_E(op)=\infty$.

  We first prove property $(L')$.
  From the definition of $T^\ast$ we immediately have $pt_E(op_{i+1})>pt_E(op_i)$ for all $1\leq i\leq k$.
  Thus, it suffices to show that $pt_E(op_i)\in I_{E}(op_i)$ for $1\leq i\leq k$.
  Since by construction $t_E\bigl(inv(op_i)\bigr)\leq pt_E(op_i)$, we only have to show that $pt_E(op_i)\leq t_E\bigl(rsp(op_i)\bigr)$.
  We prove by induction on $i$ for all $1\leq i\leq k$ the following statement:
  $pt_E(op_i)\leq t_E\bigl(rsp(op_i\bigr)$ and $pt_E(op_i)< t_E\bigl(rsp(op_j)\bigr)$ for all $i\leq j\leq k$.
  For $i=1$ the statement is true since $pt_E(op_1)=t_E\bigl(inv(op_1)\bigr)$ and any response step in $E$ occurs no earlier than $t_E\bigl(inv(op_1)\bigr)$, and only $rsp(op_1)$ can occur at the same time as $inv(op_1)$ (if $op_1$ is atomic in $E$).
  Thus, let $1<i\leq j\leq k$.
  Since $f(H)$ is a linearization of $\Gamma\bigl(H(E)\bigr)$, we know that in $E$ the invocation of $op_i$ occurs no later than the response of $op_j$, and they can occur at the same time only if $i=j$ (and $op_i$ is atomic).
  Therefore, $t_E\bigl(inv(op_i)\bigr)\leq t_E\bigl(rsp(op_j)\bigr)$, and equality holds only if $i=j$.
  Hence, if $pt_E(op_i)=t_E\bigl(inv(op_i)\bigr)$, we are done.
  So assume that $t_E\bigl(inv(op_{i-1})\bigr)<pt_E(op_i)$.
  By the induction hypothesis, and since $j>i-1$, we have $pt_E(op_{i-1})< t_E\bigl(rsp(op_j)\bigr)$, and then from the definitions of $pt_E$ and $T^\ast$, we get
  \begin{displaymath}
  pt_E(op_i)
  =
  T^\ast\bigl(pt_E(op_{i-1})\bigr)
  \leq
  \frac{pt_E(op_{i-1})+t_E\bigl(rsp(op_j)\bigr)}{2}
  < t_E\bigl(rsp(op_j)\bigr).
  \end{displaymath}

  We now show $(P')$.
  Consider two arbitrary timed executions $D,E\in\EE$ such that $D$ is a prefix of $E$.
  Let $f\bigl(H(D)\bigr)=(op_1,\dots,op_k)$.
  Since $f$ is prefix preserving, $op_1,\dots,op_k$ are the first $k$ operations in $f\bigl(H(E)\bigr)$.
  It is immediate from the inductive construction that $pt_D(op_i)=pt_E(op_i)$ for $1\leq i\leq k$, and $pt_E(op)>pt_E(op_k)=pt_D(op)$ for all operations $op\in f\bigl(H(E)\bigr)-\{op_1,\dots,op_k\}$.
  Thus, $L(D,pt_D)$ is a prefix of $L(E,pt_E)$.
\end{proof}

\begin{theorem}[Composability]\label{thm:composability}
  Let $\BB$ and $\BB'$ be disjoint sets of base objects. Further, let $O$ be a strongly linearizable implementation of some type $T$ that uses only base objects in $\BB$.
  Also, for some atomic object $B\in \BB$ let $B'$ be an object of the same type as $B$ that is implemented from objects in $\BB'$.
  Then, the implemented object $O'$ obtained from $O$ by replacing object $B$  with $B'$ is a strongly linearizable implementation of type $T$ that uses atomic base objects from the set $(\BB-\{B\})\cup \BB'$.
  \todo{the part ``\dots that uses atomic\dots'' is trivial. Omit?}
\end{theorem}
\begin{proof}
  Let $\EE_{O'}$, $\EE_{O}$, and $\EE_{B'}$ be the set of all timed executions that can occur if processes execute operations on $O'$, $O$, and $B'$, respectively.
  By the assumption, $\EE_{O}$ and $\EE_{B'}$ are strongly linearizable.
  Let $pt'_E$ for $E\in\EE_{B'}$ and $pt''_D$ for $D\in\EE_{O}$ denote mappings that map executions in $\EE_{B'}$ and $\EE_{O}$, respectively, to their strong linearization points.

  Now consider an arbitrary timed execution $E=\bigl((s_1,t_1),\dots,\bigr)\in\EE_{O'}$.
  Then $E|B'\in\EE_{\BB'}$, and thus every operation $op\in E|B'$ is associated with a strong linearization point $pt'_{E|B'}(op)$.
  For every such operation $op$ executed by some process $p$ during $E$, we do the following:
  We remove all steps from $E$, that a process $p$ executes during the interval that starts with $inv(op)$ and ends with $rsp(op)$, and if $pt'_{E|B'}(op)<\infty$, then we replace those steps  with the atomic operation $op$ at time $pt'_{E|B'}(op)$.
  Thus, in the resulting execution, $\alpha(E)$, all operations $op$ on object $B'$ are atomic.
  Also, $\alpha(E)|B'=L(E|B',pt'_{E|B'})$, and so $H(\alpha(E)|B')$ is a linearization of $\Gamma(E|B')$.
  Moreover, all steps that are not on $B'$ occur at exactly the same time in $E$ as in $\alpha(E)$, and very process executes its steps in program order.

  Thus, $\alpha(E)$ is a timed execution in $\EE_{O}$, and we can define the mapping $pt_E:\Phi\bigl(\Gamma(E)\bigr)\to\mathds{R}\cup\{\infty\}$ by
  \begin{displaymath}
    pt_E(op)=pt''_{\alpha(E)}(op).
  \end{displaymath}
    All invocation and response steps on $O$ appear at exactly the same time in $E$ as in $\alpha(E)$, i.e., $\Gamma(E)=\Gamma\bigl(\alpha(E)\bigr)$.
    Since $pt''_{\alpha(E)}(op)$ maps the operations in $\Gamma\bigl(\alpha(E)\bigr)$ to their linearization points, it is obvious that $pt_E$ satisfies property $(L')$.

    It remains to show property $(P')$.
    Let $D,E\in\EE_{O'}$, and let $D$ be a prefix of $E$.
    Then $D|B'$ is a prefix of $E|B'$, and since $pt'_{D|B'}$ and $pt'_{E|B'}$ map the operations in $D|B'$ and $E|B'$, respectively, to their strong linearization points, $\alpha(D)|B'=L(D|B',pt'_{D|B'})$ is a prefix of $\alpha(E)|B'=L(E|B',pt'_{E|B'})$.
    As a consequence, $\alpha(D)$ is a prefix of $\alpha(E)$, and so $L(\alpha(D),pt''_{\alpha(D)})=L(D,pt_D)$ is a prefix of $L(\alpha(E),pt''_{\alpha(E)})=L(E,pt_E)$.
\end{proof}

\subsection{Which Linearizable Implementations are Strongly Linearizable?}

Many linearizable implementations of shared objects from the literature are actually strongly linearizable,
and hence can be used safely in randomized algorithms, in place of their atomic counterparts.
We can identify some of these by examining the proof of linearizability of the implementation.
Such a proof typically follows one of a few general approaches.
Let $\obj$ be the implemented object and let $H$ be a history over $\obj$.

In one proof technique, for each operation $\op$ on $\obj$ a unique ``linearization point'' $pt(\op)$ is assigned,
which is a shared memory operation that occurs during the execution of operation $\op$.
A sequential history $S$ is formed by ordering the operations on $\obj$ in $H$ so that $\op_1$ precedes $\op_2$ in $S$ if (and only if)
$pt(\op_1) \prec_H pt(\op_2)$.
The construction of $S$ guarantees agreement with the ``happens before'' order of operations on \obj\ in $H$,
and so the proof obligation for linearizability is only to show that $S$ is valid for \obj.
If the linearization points are chosen so that the mapping from $H$ to $S$ also satisfies
property~(P) (see Definition~\ref{def:prefix-linearizability}), then the implementation is strongly linearizable.
Property~(P) will hold, for example, if, for any shared memory operation $\op'$ in the software method that simulates $\op$,
by the time $\op'$ is executed it is determined whether or not $\op'$ is the linearization point for $\op$.
This holds, for instance, if $pt(\op)$ is mapped statically to a particular pseudo-code statement
(that performs exactly one shared memory operation), irrespective of the schedule.

Several published implementations admit such proofs of strong linearizability, for example an obstruction-free double-ended queue \cite{her:of}
and a terminating Compare-And-Swap implementation from atomic registers \cite{ghhw:cas}.
\todo{At least for Compare-And-Swap I don't know of a simple proof. We should either remove this claim or prove it.}
The practical constructions of sets and lists such as the ``FineList'', ``OptimisticList'', and ``LazyList''
described by Herlihy and Shavit
\cite{hs:art} use fine-grained locking and are strongly linearizable for the same reason.
For non-blocking constructions of sets,
there typically remains a point in the code (often a strong synchronization operation) associated with each
 implemented ``insert''  or ``delete'' operation, which allows the proofs of linearizability to extend to
 proofs of strong linearizability (see for example the LockFreeList in \cite{hs:art}).

Many universal constructions also fall into this category because they force operations to take effect one at a time.\footnote{
In that case the linearization order follows easily from the algorithm and a formal proof
  of linearizability is often not given.}
For example, Herlihy's universal construction for wait-free shared objects \cite{herl:wait} is strongly linearizable.
In this construction, each operation applied on the shared object is represented using a ``cell'' data structure.
Processes cooperate to thread cells onto a list;
they reach agreement on the successor of each cell through a consensus object for that cell.
The linearization corresponding to the total order of the cells in this list defines a strong linearization.\footnote{
If the implementation is non-deterministic and the operation corresponding to the last cell in the list is not complete, then
          that operation is dropped from the linearization because its response may not be uniquely determined.}

Similarly, any implementation of a shared object obtained by wrapping a mutex around a sequential implementation
is strongly linearizable.
Combining this observation with known implementations of mutual exclusion using only safe registers,
yields strongly linearizable implementations of any shared object using only safe registers.
Strongly linearizable implementations with $\O(1)$ Remote Memory Reference (RMR) complexity are also possible using a queue-based mutex,
which can be constructed using atomic registers and certain Read-Modify-Write primitives such as
Fetch-And-Add or Fetch-And-Store.

In the second general proof technique,
the operations in a history $H$ are first ordered somehow into a valid sequence $S$,
and the proof obligation is to show that $S$ is consistent with the ``happens before'' order of $H$.
Such proofs typically do not directly translate to proofs of strong linearizability,
and the corresponding implementations are frequently not strongly linearizable.
Two such examples are in the introduction.
Additional examples of implementations in this category include other ``textbook'' wait-free constructions of
strong atomic registers from weaker ones.
For instance,
any strong adversary has less power against a randomized algorithm using atomic MRSW registers, than
a particular weak adversary has against the same algorithm when the MRSW registers are replaced with
Israeli and Li's linearizable construction from SRSW atomic registers \cite{il:timestamps}.
Similarly, any strong adversary has less power against a randomized algorithm using atomic MRSW registers, than
a weak adversary has against the same algorithm when the MRMW registers are replaced with
Vitanyi and Awerbuch's linearizable construction from MRSW atomic registers \cite{vit:atom}.
The impasse is not just for register-like constructions.
Herlihy and Wing \cite{her:lin} provide a linearizable implementation
of a queue object using some read-modify-write objects where the \Enqueue\ operation but not the \Dequeue\ operation is wait-free.
There is no strong linearization even for the subset of histories that occur
when \Dequeue\ is constrained to be atomic.
(Examples of all these situations are included Appendix~\ref{sec-appendix-examples}.)

%%%%% Processing file casProofs.tex
\renewcommand{\gets}{\ensuremath{:=}}
\SetKwInput{KwPrivVar}{Private variables}
\SetKwInput{KwSharVar}{Shared variables}
\SetKwInput{KwHelpFun}{Subroutines}
\SetKwFunction{NewBlockOp}{AllocBlock}
\newcommand{\algospacing}{\renewcommand{\baselinestretch}{1}}
\SetKwFunction{AlgoBrack}{()}
\SetKwFunction{AlgoBrackL}{(}
\SetKwFunction{AlgoBrackR}{)}
\newcommand{\AlgoBracks}[1]{\AlgoBrackL #1 \AlgoBrackR}
\newcommand{\alref}[1]{{\ref{#1}}}
\newcommand{\lref}[1]{line~{\ref{#1}}}
\newcommand{\lrefs}[1]{lines~{\ref{#1}}}
\newcommand{\Lref}[1]{Line~{\ref{#1}}}
\newcommand{\llref}[2]{lines~{\ref{#1}}--{\ref{#2}}}
\newcommand{\Llref}[2]{Lines~{\ref{#1}}--{\ref{#2}}}

\newenvironment{myindentpar}[1]
{\begin{list}{}
     {\setlength{\leftmargin}{#1}
     \setlength{\topsep}{1pt}
     \setlength{\itemsep}{1pt}
     \setlength{\parskip}{0pt}
     \setlength{\parsep}{0pt}}
     \item[]
}
{\end{list}}
\newcommand{\vardeclar}[1]{\begin{myindentpar}{1cm}#1\end{myindentpar}}

\subsubsection{Strong Linearizability of CAS Implementation from Atomic Registers}

We now sketch a proof of strong linearizability for the RMR-efficient implementation of Compare-And-Swap from atomic registers \cite{ghhw:cas}.
First, we describe briefly how this implementation works, letting $O$ denote the implemented CAS object.
The basic implementation supports \Read and \CAS operations, and is presented in a simplified format in Figure~\ref{fig_cas}.
(A \Write operation can be implemented in a strongly linearizable way using similar techniques \cite{golab:phd}.)
The implementation records the state of $O$ using data structures called \emph{blocks}.
Each block contains a variable $Val$ that records a state (i.e., value) of $O$.
(For any block, $Val$ inside that block is written at most once.)
The state stored in block $b$ is denoted $b.Val$.
The block that stores the latest state of $O$ is called the \emph{current} block,
  and its address is stored in a shared register $Cur$.
An operation that does not change the state of $O$ simply reads $Cur$ to identify the current block,
  and reads the state of $O$ from a shared variable in that block (\llref{cas:rb}{cas:rtf} and \llref{rd:rb}{rd:rtv}).
An operation that causes a non-trivial state change allocates a new block (\lref{cas:nb}),
  writes the new state to that block (\lref{cas:wv}), and writes the address of that block to $Cur$ (\lref{cas:wc}).
At the heart of the implementation lies a mechanism for efficient synchronization
  among operations that attempt to apply non-trivial state changes concurrently (\lref{cas:le}).
This mechanism ensures that at most one such operation succeeds
  and the others instead apply trivial state transitions.
An efficient signaling mechanism is also used between the leader and losers, as we explain shortly (\lrefs{cas:sg}, \alref{cas:wt}).

\begin{figure}[htbp]
{\algospacing
\begin{declarations}[H]
  \caption{() \protect for CAS implementation.}
  \KwSharVar{(global)
    \vardeclar{
      \begin{tabbing}
      $Cur$ -- \= stores address of current block, initially points to a block \\
      \> that records the initial value of the CAS object
      \end{tabbing}
    }}
  \KwSharVar{(per-block)
    \vardeclar{
      $Val$ -- records a state (i.e., value) of the CAS object
    }}
  \KwPrivVar{(per-process)
    \vardeclar{
      $b,b'$ -- block address \\
      $l$ -- process ID \\
      $v$ -- value of CAS object
    }}
\end{declarations}

\noindent
\begin{minipage}[t]{.64\textwidth}\ \\
\begin{procedure}[H]
  \caption{for operation CAS($X$,$Y$)}
  \SetAlgoLined
  $b \gets \Read{Cur}$\; \nllabel{cas:rb}
  $v \gets \Read(b.Val)$\; \nllabel{cas:rv}
  \uIf {$v \neq X$} {
    \Return $v$\;  \nllabel{cas:rtf}
  }
  \Else {
    $l \gets $ leader computed using leader election instance associated with block $b$\; \nllabel{cas:le}
    \uIf {$l = $ \emph{this process}} {
      $b ' \gets $ new block\; \nllabel{cas:nb}
      \Write $b'.Val \gets Y$\; \nllabel{cas:wv}
      \Write $Cur \gets b'$\; \nllabel{cas:wc}
      signal value $Y$ to losers of leader election\; \nllabel{cas:sg}
      \Return $X$\; \nllabel{cas:retls}
    }
    \Else {
      wait for signal from leader $l$\;  \nllabel{cas:wt}
      \Return value $Y$ signaled by leader\;  \nllabel{cas:retlf}
    }
  }
\end{procedure}
\end{minipage}\hfill
\begin{minipage}[t]{.30\textwidth}\ \\
\begin{procedure}[H]
  \caption{for operation Read()()}
  \SetAlgoLined
  \setcounter{AlgoLine}{17}
  $b \gets \Read{Cur}$\; \nllabel{rd:rb}
  \Return $\Read{b.Val}$\; \nllabel{rd:rtv}
\end{procedure}
\end{minipage}
}
\caption{Implementation of CAS from reads and writes. \label{fig_cas}}
\end{figure}

We now illustrate how the implementation works with a short example.
Let $s_0, s_1, s_2, ...$ denote the sequence of states
  $O$ takes on (ignoring trivial state transitions),
  and let $b_0, b_1, b_2, ...$ denote the blocks that store these states.
(The state of a CAS object is simply the value it stores.)
Suppose that the current block is $Cur = b_2$ and processes $p_1,p_2,p_3,p_4$ invoke
  the following operations concurrently:
  $p_1$,$p_2$ and $p_3$ invoke $\CAS(s_2, X_1)$, $\CAS(s_2, X_2)$ and $\CAS(s_2, X_3)$ respectively,
  where $s_2 \not\in \left\{X_1, X_2, X_3\right\}$; and
  $p_4$ invokes $\CAS(Y, Z)$ for some $Y \not\in \left\{s_2, X_1, X_2, X_3\right\}$ and arbitrary $Z$.
Process $p_4$'s operation is dealt with easily; upon discovering that the state $s_2$ stored in block $Cur = b_2$
  is different from $Y$, $p_4$'s $\CAS(Y, Z)$ does not change the state of $O$ and returns $s_2$ (\lref{cas:rtf}).
In contrast, each of $p_1,p_2,p_3$ has a chance to change the state of $O$, although only one of them may actually do so.
Thus, $p_1,p_2,p_3$ elect a leader, say $p_2$, which allocates a new block $b$ and records $X_2$ in it (\llref{cas:nb}{cas:wv}).
(A separate ``instance'' of leader election is used to synchronize each non-trivial state change, and we associate
   such instances with blocks on a one-to-one basis.
In this example, $p_1,p_2,p_3$ use the instance associated with block $b_2$.)
Process $p_2$ then sets $Cur = b$ (\lref{cas:wc}) and its $\CAS(s_2, X_2)$ returns $s_2$ (\lref{cas:retls}).
Thus, $b_3 = b$ and $s_3 = X_2$ hold, and the state of $O$ changes from $s_2$ to $s_3$.
Meanwhile, $p_1$ and $p_3$ wait until $p_2$ has written $Cur$ (\lref{cas:wt}),
  and their operations return $s_3$ (\lref{cas:retlf}).
(The linearizability of the implementation depends crucially on $p_1$ and $p_3$ waiting for $p_2$ in this situation.)
Thus, their (``failed'') \CAS operations cause trivial state transitions from $s_3$ back to $s_3$, and appear to take
   effect just after $p_2$'s (``successful'') \CAS.

To show strong linearizability, we will use timed executions (see Section~\ref{sec:timed}).
For any history $H$ of the implementation, define a corresponding timed execution $E$
  where the $i$'th step occurs at time $i$.
Let $\EE$ denote the set of such timed executions.
For any $E \in \EE$, we define the mapping $pt_E:\Phi\bigl(\Gamma(E)\bigr)\rightarrow\mathds{R}\cup\{\infty\}$
  as follows.
If $\op$ is a \Read operation that invokes a read of $Cur$ at time $t$, then $pt_E(\op) = t$, otherwise
   $\op$ is pending in $E$ and $pt_E(\op) = \infty$.
If $\op$ is a $\CAS(X, X)$ operation for some $X$ (i.e., the comparison value is the same as the new value),
  we treat it just like a \Read.
If $\op$ is a $\CAS(X, Y)$ operation for some $X$ and $Y \neq X$, then $pt_E(\op)$ depends on the execution
   path of $\op$.
If $\op$ is pending and does not read $Cur$ at all, then $pt_E(\op) = \infty$.
If $\op$ invokes a read of $Cur$ at time $t$, and the state of $\O$ in the block read is not $X$, then
   $\op$ is a ``failed'' \CAS and $pt_E(\op) = t$.
If $\op$ invokes a read of $Cur$ at time $t$, and the state of $\O$ in the block read is $X$, then
   $\op$ may succeed or fail depending on the outcome of leader election
   (i.e., the ``instance'' of leader election associated with the block whose address $\op$ reads from $Cur$).
If the outcome of leader election is not ``decided'' at the end of $E$ (i.e., different extensions of $E$
   may lead to different outcomes), or the outcome is decided
   but the leader's operation has not yet overwritten $Cur$, then $pt_E(\op) = \infty$.
(In this case $\op$ is pending because the implementation ensures that the leader writes $Cur$
   before $\op$ terminates, even if $\op$ is not the leader's operation.)
Otherwise the write to $Cur$ by the leader is invoked at some time $t$.
If $\op$ is the leader's operation then $pt_E(\op) = t$.
Finally, if $\op$ is not the leader's operation, then $pt_E(\op) = t + \eps$ where
   $0 < \eps < 1$ is an arbitrary constant unique for each process.

It is straightforward to verify that for each $E \in \EE$, $L(E,pt_E)$ is a linearization of $\Gamma(E)$.
Next, consider Definition~\ref{def:slm}.
Property $(L')$ follows easily for any operation $\op$ whose linearization point is
   the time of a step that $\op$ itself takes, which is always between the
   times of $\inv{\op}$ and $\rsp{\op}$.
In all remaining cases, the linearization point is of the form $t + \eps$ where $t$ is the time
   of a step that is not part of $\op$ but nevertheless occurs after $\inv{\op}$ and before $\rsp{\op}$.
Since the times at which steps occur are integer-valued by construction of $E$, and since $0 < \eps < 1$, it follows
   that $t + \eps$ is also between the times of $\inv{\op}$ and $\rsp{\op}$, as wanted.
For property $(P')$, consider histories $D,E \in \EE$ such that $D$ is a prefix of $E$.
It suffices to consider the case when $E$ extends $D$ by one step, say $s$ at time $t$,
   and show that $L(D, pt_D)$ is a prefix of $L(E,pt_E)$.
If $s$ is not a step that reads or writes the shared variable $Cur$, then $L(E,pt_E) = L(D, pt_D)$.
If $s$ reads $Cur$ and is part of some operation $\op$, then either $pt_E(\op) = \infty$ and
   $L(E,pt_E) = L(D, pt_D)$ (i.e., $\op$ is pending and trying to cause a non-trivial state change),
   or else $L(E,pt_E)$ is an extension of $L(D, pt_D)$ by an invocation/response pair for $\op$
   (i.e., $\op$ causes a trivial state change).
Finally, if $s$ writes $Cur$ and is part of operation $\op$, then $L(E,pt_E)$ is an extension of
   $L(D, pt_D)$ by an invocation/response pair for $\op$, followed by zero or more
   invocation/response pairs for operations whose linearization points are of the form $t + \eps$,
   $0 < \eps < 1$.

%%%%% Done processing casProofs.tex

%%%%% Done processing strongAdversary.tex

%%%%% Processing file weakAdversary.tex
\section{Weak Adversaries}
\label{weakAdversary.sec}
In this section we show that it is impossible to strengthen linearizability in a way that limits the power of
a weak adversary to influence the result of an execution in the same sense as strong linearizability limits the
power of the strong adversary, when implementations are obtained only from atomic registers and load-linked/store-conditional.
To prove this, we consider a particular algorithm that uses shared objects of a ``strong counter'' type.
The state of this type is an integer (initially 0) and the operations supported are \FAI{} and \FAD{}.
These operations increment and decrement the counter, respectively, and also return the prior value of the counter.
In the executions we will consider, operations will be invoked in such a way that the counter's value
is always in $[0, n]$ where $n$ is the maximum number of processes.

In Figure~\ref{fig:loadbalance}, we present a simple algorithm that uses $\sqrt{n}$ strong counters.
(Throughout this section we assume that $n$ is a perfect square.)
Each process chooses one of the counters uniformly at random, then calls \FAI{}, and finally calls \FAD{}.
Now fix a weak adversary $\AA$ and an integer $K_{max}$.
Let $H$ be a random history obtained by a run of $\LoadBalance{}$ scheduled by $\AA$. % and using the random coin flip vector $\tc\in\Omega^\infty$.
For each process $p\in\PP$ let $X_{\AA,p}$ be the random variable defined as follows:
If during $H$ the maximum point contention%
\footnote{The maximum point contention is the maximum number of processes that at the same point in time have called \LoadBalance{} but not yet finished that function call, where the maximum is taken over all points in time.}
is at most $K_{max}$ and $p$'s \FAI{} terminates, then $X_{\AA,p}$ is the value returned by that \FAI{}.
If $p$'s \FAI{} does not terminate during $H$, or if the point contention exceeds $K_{max}$, then define $X_p=0$.
We are interested in the maximum expectation of $X_{\AA,p}$  over all processes, i.e.,
\begin{displaymath}
  \Phi(\AA):=\max\bigl\{\Exp{X_{p.\AA}}\,|\,p\in\PP\bigr\}.
\end{displaymath}
It is not hard to see that if the \FAI{} operation is atomic, then for all weak adversaries $\AA$, $\Phi(\AA)\leq K_{\max}/\sqrt{n}$.
In particular, if $K_{\max}=\Theta(\sqrt{n})$, then $\Phi(\AA)=O(1)$.

  \begin{figure}[htbp] % using package float
 \begin{function}[H]
    \SetKwInput{SharDat}{Shared data}
    \caption{LoadBalance()()}
    \SharDat{$\sqrt{n}$ shared strong counters, $F_0,\dots,F_{\sqrt{n}-1}$}
    Choose $i\in\{0,\dots,\sqrt{n}-1\}$ uniformly at random.\label{LoadBalance:random_choice}\;
    $x:=F_i$.\FAI{}\;
    $F_i$.\FAD{}\;
    \Return{$x$}\;
  \end{function}
  \caption{The load balancing algorithm.}\label{fig:loadbalance}
  \end{figure}

On the other hand, we will show that if the \FAI{} operations are based on a ``natural'' terminating (or lock-free) implementation that uses only \read, \xwrite, and \LL/\SC operations, then for $K_{\max}=\Theta(\sqrt{n})$ there exists a weak adversary $\AA$ such that $\Phi(\AA)=\Omega(K_{\max})=\Omega(\sqrt{n})$.
An implementation of a type $\tau$ is natural if each instance $O$ of the implemented object has its own set $\mathcal{B}_O$ of base objects such that an operation on the instance $O$ accesses only base objects in $\mathcal{B}_O$.
Essentially all linearizable implementations of objects are natural, as otherwise the composition of multiple objects would effectively create a single new object.

Strong counters can be implemented with various progress properties from atomic \read, \xwrite, and \LL/\SC operations.
For example, a deterministic wait-free implementation can be obtained using Herlihy's universal construction \cite{herl:wait}, with consensus objects simulated in a straightforward way from \LL/\SC.
Given only \read and \xwrite, a deterministic terminating implementation is possible using Yang and Anderson's mutual exclusion algorithm \cite{yang:fast}.
(Alternately, one can simulate \LL/\SC in the wait-free implementation using \read and \xwrite \cite{ghhw:cas}, which yields a terminating strong counter implementation from \read and \xwrite only.)
All these implementations are natural.

We now present our main result, stating that a weak adversary gains additional power against algorithm~\LoadBalance\
no matter how the strong counters are implemented from \read, \xwrite, and \LL/\SC.

\begin{theorem}\label{thm:weak-adversary-combined-bound}
  Let $K_{\max}=\bigl\lceil(1+\delta)\sqrt{n}\bigr\rceil$, $\delta>0$.
  If the algorithm in Figure~\ref{fig:loadbalance} (\LoadBalance)
  is used with atomic strong counters, then
  for any weak adversary $\AA$,
  $\Phi(\AA)<1+\delta=O(1)$ holds.
  On the other hand, if the algorithm~\LoadBalance\ is used with a linearizable, natural, possibly randomized, terminating (or lock-free)
  implementation of strong counters from atomic \read, \xwrite, and \LL/\SC operations, then
  there exists a weak adversary $\AA$, such that
  $\Phi(\AA)=\Omega(K_{max})=\Omega(\sqrt{n})$.
\end{theorem}
\noindent The full proof of the theorem is quite complicated and given in the extended version of this paper.
Here we sketch the main idea.

\begin{proof}[Proof sketch]
For the upper bound, observe that at the point in time when a process $p$ makes its random choice, the expected value of the counter it chooses is at most $k/\sqrt{n}$, where $k$ is the number of processes currently active.
Since the weak adversary cannot intervene between $p$'s random choice and $p$'s \FAI{}, that operation's return value is at most $(k-1)/\sqrt{n}$.

For the lower bound, we construct for each process $p$ a weak adversary $\AA_p$ that tries to ``fool'' $p$.
We fix the coin flips that processes receive arbitrarily.
Then we choose one process $p$ at random, use $\AA_p$ for the scheduling, and obtain that the expected return value of $p$'s \FAI{} is $\Omega(k)$, where $k$ is the number of processes that randomly chose the same counter as $p$.
Also, point contention is at most $k+1$.
If coin flips are uniformly random, then $k$ is highly concentrated around $\sqrt{n}$, so with high probability it will be $\Omega(\sqrt{n})$ but also not exceed $K_{max}$.

So the main goal of the adversary $\AA_p$ is to make $p$'s \FAI{} call return a value with expectation $\Omega(k)$.
We achieve this as follows.
First, we let process $p$ take one step, which reveals the index $i^\ast$ of the counter it is using.
Then we schedule each process $q\neq p$, one after the other.
If $q$ accesses a different counter than the one chosen by $p$, we let $q$ run solo until it finishes the algorithm (after that $q$ does not contribute to point-contention anymore).
If $q$ also chooses the counter with index $i^\ast$, then we let it take exactly one step and then stall it.
This way, eventually all $k$ processes that chose the same counter as $p$ are stalled, while all other processes are finished.
Moreover, the maximum point contention encountered is at most $k+1$.

Now we partition the set of stalled processes into three sets.
Let $\op_q$ be the first (and only) operation process $q$ executed so far.
The set $\QQ$ contains all processes $q$, where $\op_q$ is not a $\xwrite$.
The set $\VV$ contains all processes $q$, where $\op_q$ is a $\xwrite$ to a register $R$ to which no other process writes in its first step.
Finally, $\WW$ is the set of all remaining processes, which write to some register that is also written by another process.
Note that the sets $\QQ$, $\VV$, and $\WW$ are uniquely determined by the first steps executed by all processes, and thus by the coin flips processes use.
But they are independent of the choice of $p$.

We make use of a positive correlation between the size of each set, and the probability with which $p$ is from that set.
Suppose $\QQ$ is large, i.e., $|\QQ|\geq k/3$.
Then the probability that $p\in\QQ$ is at least $1/3$.
Moreover, given that $p\in\QQ$, $p$'s first operation $\op_p$ is not a $\xwrite$, and thus leaves no ``trace'' (note that if it is a \SC, then it fails).
In this case, we stall $p$ and let all other $k-1$ processes finish their $F_{i^\ast}.\FAI{}$ operation.
After that $F_{i^\ast}$ has value $k-1$, so that when we finally let $p$ finish its \FAI{}, its return value is $k$.

Now suppose $|\VV|\geq k/3$, so the probability that $p\in\VV$ is at least $1/3$.
If $p\in\VV$, we let all processes in $\VV$ run in a round-robin fashion in some predetermined order, until they have finished their algorithm.
It can be argued that the choice of $p$ among processes in $\VV$ has no influence on what processes in $\VV$ observe in the resulting execution.
Hence, the \FAI{} operations of all processes in $\VV$ return distinct values.
Given that $p\in\VV$ is chosen uniformly at random, $p$'s \FAI{} return value has an expectation of at least $(|\VV|-1)/2=\Omega(k)$.

Finally, suppose $|\WW|\geq k/3$.
If $p\in\WW$, then $\op_p$ is a \xwrite that was overwritten by the operation $\op_q$ of some other process $q$ (recall that $p$'s \xwrite was scheduled before any other process took a step).
Moreover, since the first operation of all other processes in $\VV\cup\WW$ is a $\xwrite$, none of these processes can have ``seen'' $p$ before $p$ was overwritten.
Any process that may have seen $p$ cannot have executed a \xwrite in its first step, so it didn't become ``visible'' itself.
We let all processes in $\VV\cup\WW-\{p\}$ finish their \FAI call by scheduling them in a round-robin fashion.
They can only see themselves, and not any process outside of this set, so they will increase the value of the counter to at least $|\VV\cup\WW|\geq k/3$.
When we run $p$ afterward, its \FAI{} call must return a value of $\Omega(k)$.
\end{proof}

For the complete proof of Theorem~\ref{thm:weak-adversary-combined-bound} we rely on
Theorems~\ref{thm:weak-adversary-atomic} and \ref{thm:weak-adversary-impossiblity} below.
\begin{theorem}\label{thm:weak-adversary-atomic}
  If Algorithm~\ref{fig:loadbalance} (\LoadBalance) is used with atomic strong counters, then
  for any weak adversary $\AA$ and any integer $K_{max}$
  \begin{displaymath}
     \Phi(\AA)\leq\frac{K_{\max}-1}{\sqrt{n}}.
  \end{displaymath}
\end{theorem}
\begin{proof}
  Let $m=\sqrt{n}$ and let $\AA$ be an arbitrary weak adversary.
  Fix an arbitrary process $p$.
  Since $X_{\AA,p}$ is 0 if $K>K_{\max}$, it suffices to show that $\CondExp{X_{\AA,p}}{K}\leq(K-1)/m$.

  Consider the system configuration immediately before $p$ makes its random choice in line~\ref{LoadBalance:random_choice}.
  Suppose that in this configuration the counter $F_j$, $0\leq j<m$, has value $b_j$.
  Then clearly there are at least $b_0+\dots+b_{m-1}+1$ processes active (including $p$).
  Thus, given $K$, we have $b_0+\dots+b_{m-1}+1\leq K$.
  Hence, when $p$ makes its random choice, the expected value of the counter chosen by $p$ is
  \begin{displaymath}
    \sum_{0\leq j< m}\frac{b_j}{m}\leq \frac{K-1}{m}.
  \end{displaymath}
  Since the adversary cannot intervene between $p$'s random choice and $p$'s following \FAI{} operation, the expected return value of that operation is at most $(K-1)/m$,
  as wanted.
\end{proof}

Next, we consider the case when Algorithm~\ref{fig:loadbalance} is used with linearizable implementations  of strong counters from atomic \read, \xwrite and \LL/\SC operations.
Our analysis applies to any terminating or lock-free (hence wait-free) implementation of strong counters from \read, \xwrite and \LL/\SC
that is \emph{natural}.

\begin{theorem}\label{thm:weak-adversary-impossiblity}
  Let $K_{\max}=\bigl\lceil(1+\delta)\sqrt{n}\bigr\rceil$, for some arbitrarily small $\delta>0$.
  Consider Algorithm~\ref{fig:loadbalance} (\LoadBalance) used with a linearizable, natural, possibly randomized, terminating (or lock-free)
  implementation of strong counters from atomic \read, \xwrite, and \LL/\SC operations.
  Then there exists a weak adversary $\AA$, such that
  \begin{displaymath}
    \Phi(\AA)=\Omega(K_{max})=\Omega(\sqrt{n}).
  \end{displaymath}
\end{theorem}

To prove Theorem~\ref{thm:weak-adversary-impossiblity}
first we specify the adversary $\AA$ (in Subsection \ref{adversary.subsec}),
and second present the analysis that bounds $\Phi(\AA)$ from below (Subsection \ref{analysis.subsec}).

For the purpose of the following definition, it is convenient to assume w.l.o.g.\ that registers store pairs of values, where the second component of such a pair stores which process changed the register last.
We can achieve this by having processes write their ID into the second component of the pair with each successful \xwrite or \SC operation.
We say, process $q$ \emph{marks} the register with its ID.
If later on some other process $z\neq q$ writes to the same register, or performs a successful \SC operation on that register, $q$'s mark gets replaced by $z$'s mark.

We say that in some configuration $C$ process $p$ is \emph{visible}, if in $C$ some register is marked by $p$.
During a history $H$ a process $q$ \emph{sees} process $p\neq q$, if
\begin{itemize}
 \item either $q$ reads or performs an \LL operation on a register marked by $p$, or
 \item $q$ executes an \SC operation on some register $R$, at a point in time when $R$ is marked by $p$, and after $q$ has executed an \LL operation on $R$.
 (If this happens, then $q$'s \SC operation fails.)
\end{itemize}

\begin{observation}\label{obs:indistinguishable}
If $H$ is a history where no process in $P\subseteq\PP$ sees a process in $\overline{P}$, then $H|P$ is indistinguishable from $H$ to all processes in $P$.
\end{observation}

\subsection{\mathversion{bold}Specification of the Adversary $\AA_p$}
\label{adversary.subsec}
First we describe for each process $p$ some weak adversary $\AA_p$ that is trying to ``fool'' the process $p$
For the analysis we will then choose $p$ at random.

Adversary $\AA_p$, $p\in\PP$, is defined as follows:
We let $p$ take steps until $p$ has performed its first shared memory access (during its \FAI{} operation).
Note that this access reveals the index $i^\ast$ of the counter $F_{i^\ast}$ that $p$ chose.
After that we stall $p$.

Next we consider each of the remaining processes $q\in\PP-\{p\}$, one after the other, in the order of their IDs.
We let $q$ take steps until it has executed its first shared memory access and thus revealed its choice of a counter $F_{i_q}$.
If $i_q\neq i^\ast$, we let $q$ run solo, until it has finished its run of \LoadBalance completely.
Otherwise, we stall $q$.
This is continued until each processes in $\PP$ either has finished its \LoadBalance algorithm or is stalled.

Let $H'$ be the execution obtained so far, and $C$ the system configuration at the end of $H'$.
Further, let $\PP_j$, $0\leq j<m$, denote the set of the processes that selected the counter $j$ during $H'$.
Note that in $H'$, all processes in $\PP-\PP_{i^\ast}$ finish their \LoadBalance{} call, and all processes in $\PP_{i^\ast}$ (including $p$) have executed exactly one shared memory access.
We distinguish two cases.

\paragraph{\mathversion{bold}Case~1: In configuration $C$, $p$ is visible.}
Let $\WW$ be the set of processes in $\PP_{i^\ast}$ whose first shared memory access was a write.
Note that among all processes in $\PP_{i^\ast}$, only processes in $\WW$ could have changed the value of a shared register during $H'$  (all \SC operation must have failed, because they weren't preceded by a \LL).
Hence, since $p$ is visible in $C$, $p\in\WW$.
We let all processes in $\WW$, including $p$, take steps in a round-robin way, ordered by their IDs.
Since no process in $\WW$ can ever see a process not in $\PP_{i^\ast}$, all those processes eventually make progress and finish their \FAI{} call.
When a process finishes its \FAI{} call, it is stopped and the scheduling continues with the remaining processes in $\WW$, until all processes are stopped.

\paragraph{\mathversion{bold}Case~2: In configuration $C$, $p$ is not visible.}
Although $p$ is not visible in configuration $C$, it is possible that some processes see $p$ during $H'$.
This can happen, for example, if $p$ writes to a register $R$, then some process $q$ reads $R$, and finally a third process writes to $R$, overwriting whatever $p$ has written.
Let $\SS$ be the set that contains $p$ and all processes $q\neq p$ that see $p$ during $H'$.
We let all processes in $\PP_{i^\ast}-\SS$ take steps in a round-robin way, ordered by their IDs.
Since each process $q\in\SS-\{p\}$ sees $p$ when $q$ executes its first shared memory operation, that operation must be either a $\LL$, a $\read$, or a failed $\SC$.
Hence, no process in $\SS$ is visible in configuration $C$ or has been seen by a process in $\PP_{i^\ast}-\SS$.
Thus, all processes in $\PP_{i^\ast}-\SS$ eventually make progress and finish their \FAI{} call when scheduled in a round-robin fashion.
When this happens for a process, that process is stopped, and the scheduling continues with the other processes, until all processes are stopped.
Finally, we let $p$ run until it has finished its \FAI{} call.
(Since no other process $q\in\SS$ is visible, and $p$ has not seen a process from $\SS$, $p$'s \FAI{} call will eventually terminate.)

\subsection{Analysis}
\label{analysis.subsec}
For our analysis we assume that $n$ is a perfect square and let $m=\sqrt{n}$.
Later in the analysis, we will fix a sequence of coin flips in such a way that each process gets the same coin flips during every run of the algorithm, no matter what the scheduling is.
Let $\tc$ denote such a sequence of coin flips.
(Note that $\tc$ is not the same as the coin-flip vectors $\vc$ used in previous sections.)
However, for any fixed adversary $\AA$ and a fixed process $p$, each coin-flip sequence $\tc$ uniquely determines a random history $H$ and we can choose $\tc$ uniformly  at random to determine $X_{\AA,p}=X_{\AA,p}(\tc)$.

For the analysis, we will let the adversary $\AA_p$ do the scheduling in order to try to ``fool'' process $p$.
Since we don't know which adversary is the best for the purpose of a lower bound, we choose $p\in\PP$ at random.
To emphasize that $p$, $\AA$ and $\tc$ are random (although there will be a dependence between $p$ and $\AA$), we now denote $X_{\AA,p}(\tc)$ by the random variable $X=X(\AA,p,\tc)$.
We prove the following:
\begin{lemma}\label{lem:weak-adversary-main}
  For $\tc$ and $p\in\PP$ chosen uniformly at random,
  \begin{displaymath}
  \Exp[p,\tc]{X(\AA_p,p,\tc)}=\Omega(\sqrt{n}).
  \end{displaymath}
\end{lemma}

\noindent By averaging over all processes $p\in\PP$, there exists a process $p$ such that for random $\tc$
\begin{equation}\label{eq:weak_after_averaging}
  \Exp[\tc]{X(\AA_p,p,\tc}=\Omega(\sqrt{n}).
\end{equation}
Since
\begin{displaymath}
 \Phi(\AA_p)=\max_{q\in\PP}\bigl(\Exp[\tc]{X(\AA_p,q,\tc)}\bigr)\geq\Exp[\tc]{X(\AA_p,p,\tc)}
\end{displaymath}
Theorem~\ref{thm:weak-adversary-impossiblity} follows immediately.
Thus, it suffices to prove Lemma~\ref{lem:weak-adversary-main}.

In order to prove the lemma, we first consider an arbitrarily fixed coin-flip sequence $\tc$, and we choose only $p\in\PP$ uniformly at random.
Let $H=H(p)$ be the random history obtained if adversary $\AA_p$ schedules a run of $\LoadBalance{}$, for the fixed coin flips $\tc$.
By construction of $\AA_p$, during $H$ process $p$'s \FAI{} call returns.
We write $L=L(p)$ for the random variable that denotes the return value of that $\FAI{}$-call.
(Note that the value of $X$ can be 0 even if $L>0$, because in $H$ contention might exceed $K_{max}$.)
Since $\tc$ is fixed, each process gets the same results from its coin-flips, independently of the scheduling, and thus independently of the choice of $p$.
Hence, the set $\PP_i$, $0\leq i\leq m$, of processes that choose counter $F_i$ is fixed, too.
Moreover, since we choose $p\in\PP$ at random, the counter chosen by $p$ is determined by which set $\PP_i$ process $p$ is taken from.

In the following, we analyze the expectation of $X$ conditioned under the event that $p$ is in $\PP_{j}$ for an arbitrary index $j$.
More precisely, we show the following:
\begin{lemma}\label{lem:L_vs_Pj}
  For any fixed $\tc$, any $0\leq j<m$, and for $p\in\PP_j$ chosen uniformly at random,
  \begin{displaymath}
  \Exp{L}=\Omega(|\PP_j|).
  \end{displaymath}
\end{lemma}
For the ease of notation, we prove the statement for $j=0$.
Analogous arguments prove the statement for an arbitrary choice of $j$,
  and so we lose no generality.

Let $H'=H'(p)$ be the prefix of $H$ that ends when configuration $C=C(p)$ is reached (i.e., when all processes in $\PP_0$ are stalled after they have executed their first shared memory access), and let $H''=H''(p)$ the suffix such that $H=H'\circ H''$.
Let $\op_q$ be the shared memory operation process $q\in\PP$ performs during $H'$.
Let $\VV$ be the set of processes in $\PP_{0}$, that are visible in configuration $C$ and whose first shared memory operation during $H'$ is not a \xwrite to a register that was previously written by some other process.
I.e., the operation $\op_q$ of a process $q\in\VV$ is a \xwrite to some register $R$, such that no other process $q'$ writes to $R$ during $H'$.

Let $\EE$ be the event that $p$ is visible in configuration $C$.
Since $p$ is the first process to execute a shared memory access, $\EE$ occurs if and only if $p\in\VV$.
Note that if event $\EE$ occurs, then the adversary acts as described under Case~1, otherwise as in Case~2.
In the following we give lower bounds on the conditional expectation of $L$ for both cases.
\subsubsection{Analysis of Case~1}
\begin{claim}\label{clm:lb_p_visible}
  The set $\VV$ is independent of $p$, and
  \begin{math}
    \CondExp{L}{\EE}\geq (|\VV|-1)/2.
  \end{math}
\end{claim}
\begin{proof}
Recall that in Case~1 we assume that event $\EE$ has occurred, and so $p\in\VV$.
On the other hand, process $q$ is in $\VV$ if and only if $\op_q$ is a \xwrite to some register $R_q$ and no other process writes to $R_q$ during $H'$.
  Whether a process writes in the first step, and if so, to which register it writes, depends only on $q$'s coin-flips, i.e., on $\tc$, but not on the choice of $p$.
Thus, $\VV$ is independent of $p$.

Now let $C_R=C_R(p)$ be the (random) configuration of the shared memory at the end of $H'$.
  (Note that process states are not captured by $C_R$.)
Note that the position in $H'$ of the event where a process $q \in \VV$ executes its shared memory operation $\op_q$ on register $R_q$
  has no effect on the configuration $C_R$ because $q$ is the only process in $\VV$ writing to $R_q$.
Also, given that $p\in\VV$, in $H'$ the relative order of steps by processes that are not in $\VV$ is independent of $p$.
  Hence, given that $p\in\VV$, the register configuration $C_R$ is independent of $p$.
  Formally,
  \begin{equation}\label{eq:register-configur-indep}
    \forall q,q'\in\VV:\ C_R(q)= C_R(q').%
    \footnote{%
    With $C_R(q)$ we denote the register configuration obtained if the random process $p=q$.
    Similarly we write $H(q)$, $H'(q)$, $H''(q)$, and $C(q)$.}
  \end{equation}

  Recall that $\WW$ is the set of processes in $\PP_0$ whose first shared-memory step is a \xwrite.
  Clearly, a process in $\WW$ cannot see any other process during $H'$, because it performs no more steps after its \xwrite.
  Hence, given $p\in\VV\subseteq\WW$, what processes in $\WW$ observe during $H'$ is independent of $p$.
  Thus,
  \begin{equation}\label{eq:H'-indep}
    \forall q,q'\in\VV:\ H'(q)\sim_{\WW}H'(q').
  \end{equation}
  Hence, at the end of $H'$ the state of a processor in $\WW$ is independent of the choice of $p\in\VV$.
  From this and (\ref{eq:register-configur-indep}),
  \begin{equation}\label{eq:configur-indep}
    \forall q,q'\in\VV:\ C(q)\sim_{\WW}C(q').
  \end{equation}
  In the second part of the scheduling by process $\AA_p$, which determines $H''$, only processes in $\WW$ get scheduled in a round-robin way, ordered by their IDs.
  Given $p\in\VV\subseteq\WW$, that scheduling is independent of $p$.
  Since, given $p\in\VV$, processes in $\WW$ cannot distinguish between the possible configurations $C$ obtained at the end of $H'$, they cannot distinguish between the histories $H''(p)$, $p\in\VV$, either.
  I.e.,
  \begin{displaymath}
    \forall q,q'\in\VV:\ H''(q)\sim_{\WW}H''(q').
  \end{displaymath}
  From this and (\ref{eq:H'-indep}), we get
  \begin{displaymath}
    \forall q,q'\in\VV:\ H(q)\sim_{\WW}H(q').
  \end{displaymath}
  As a consequence, given $p\in\VV$, the value returned by the \FAI{} of process $q\in\VV$ is independent of $p$.
  Since all processes in $\VV$ finish their \FAI{} during $H$, but none of them starts its $\FAD{}$, the values processes in $\VV$ receive from their $\FAI{}$ calls are all distinct non-negative integers from a (fixed) set $S=\{s_1,\dots,s_{|\VV|}\}$.
  Thus, given that $p$ is distributed uniformly over $\VV$, the return value of $p$'s \FAI{} is distributed uniformly over $S$.
  The expectation of that value is at least $(|\VV|-1)/2$.
\end{proof}

\subsubsection{Analysis of Case~2}
\begin{claim}\label{clm:lb_helper}\samepage
  Let $H$ be an arbitrary history of Algorithm~\ref{fig:loadbalance} (\LoadBalance),
  and let $P\cup\{p\}$, $p\not\in P$, be a set of processes that call $F_i.\FAI{}$, $i\in\{0,\dots,m-1\}$.
  Suppose the following hold for $H$:
  \begin{enumerate}
    \item[(a)] the processes in $P\cup\{p\}$ all finish their $F_i$.\FAI{},
    \item[(b)] no process in $\PP$ calls $F_i$.\FAD{}, and
    \item[(c)] none of the processes in $P$ sees a process in $\overline{P}$.
  \end{enumerate}
  Then process $p$'s $F_i.\FAI{}$ returns a value of at least $|P|$.
\end{claim}
\begin{proof}
  By (a) and (b) all processes in $P\cup\{p\}$ receive distinct values from their \FAI{} calls.
  By (c) and Observation~\ref{obs:indistinguishable}, $H$ and $H|P$ are indistinguishable to the processes in $P$, and so the return value of a \FAI{} call by a process $q\in P$ is at most $|P|-1$.
  Hence, the set of $\FAI{}$ call return values of processes in $P$ is exactly $\{0,\dots,|P|-1\}$, and so $p$'s \FAI{} call must return a value of at least $|P|$.
\end{proof}

Now let $\QQ$ be the set of processes in $\PP_0$ whose first shared-memory access is a \read, a \LL, or a \SC (which fails).
Note that $\QQ$ is independent of the random choice of $p\in\PP_0$.
Recall that $\SS$ is the set that contains $p$ and all processes that see $p$ during $H'$.
Since the \FAI{} implementation is natural, $\SS\subseteq\PP_0$.
Let $S=S(p)=|\SS|$.

\begin{claim}\label{clm:lb_p_not_visible}
  If event $\EE$ does not occur, then $L\geq |\PP_0|-S\geq |\PP_0|-|\QQ|$.
\end{claim}
\begin{proof}
  If event $\EE$ does not occur, then $p$ is not visible in configuration $C$.
  By definition of the adversary, during $H''$ only processes in $(\PP_0-\SS)\cup\{p\}$ take steps, and finish their \FAI{} call but don't start their \FAD{} call.

  A process in $\SS-\{p\}$ does not ``leave a trace'' during $H$, i.e., it is never visible: Since it sees $p$ its only operation occurs during $H'$, and that operation is either a \read, a $\LL$, or a failed $\SC$.
  In particular, during $H$, no process can see a process in $\SS-\{p\}$.
  In addition, no process $q\in\PP_0-\SS$ sees $p$ during $H$:
  It cannot see $p$ during $H'$, or else it would be in $\SS$, and it cannot see $p$ during $H''$, because $p$ remains invisible until $q$ has finished its \FAI{} call and is stopped.
  Moreover, since the implementation is natural, a process in $\PP_0$ cannot see a process in $\overline{\PP_0}$.
  To conclude, no process in $\PP_0-\SS$ sees any process in $\SS\cup\overline{\PP_0}=\overline{\PP_0-\SS}$.
  Thus, the conditions of Claim~\ref{clm:lb_helper} are satisfied for $P=\PP_0-\SS$, and so $L\geq |\PP_0-\SS|=|\PP_0|-S$.
  Because $\op_q$ is a \read, a \LL, or a failed \SC for every process $q\in\SS$, we have $S\leq |\QQ|$ and thus $|\PP_0|-S\geq |\PP_0|-|\QQ|$.
\end{proof}

\subsubsection{Putting Things Together}
\begin{proof}[Proof of Lemma~\ref{lem:L_vs_Pj}]
  W.l.o.g.\ assume $j=0$.
  Which operation a process in $\PP_0$ executes in its first shared memory step is independent of the choice of $p\in\PP_0$.
  Thus, since $p$ is uniformly distributed in $\PP_0$, we have
  \begin{equation}\label{eq:prob_p_in_Q}
    \Prob{p\in\QQ}=\frac{|\QQ|}{|\PP_0|}.
  \end{equation}
  Similarly, since by Claim~\ref{clm:lb_p_visible} $|\VV|$ is independent of $p$, and since event $\EE$ occurs if and only if $p\in\VV$,
  \begin{equation}\label{eq:prob_p_in_V}
    \Prob{\EE}=\Prob{p\in\VV}=\frac{|\VV|}{|\PP_0|}.
  \end{equation}
  Applying Claim~\ref{clm:lb_p_visible} and (\ref{eq:prob_p_in_V}) we obtain
  \begin{equation}\label{eq:L_p|p_in_V}
    \CondExp{L}{\EE}\cdot\Prob{\EE}
    \geq
    \frac{|\VV|-1}{2}\cdot\frac{|\VV|}{|\PP_0|}.
  \end{equation}
  If $p\in\QQ$, then no process sees $p$ during $H'$, i.e., $S=0$.
  Moreover, if $p\in\QQ$, then $p$ is not visible in configuration $C$ and thus $\EE$ does not occur.
  Consequently, by Claim~\ref{clm:lb_p_not_visible}, $L\geq|\PP_0|$, and thus applying
  (\ref{eq:prob_p_in_Q}),
  \begin{equation}\label{eq:L_p|p_in_R}
    \CondExp{L}{p\in\QQ}\cdot\Prob{p\in\QQ}
    \geq
    |\PP_0|\cdot\frac{|\QQ|}{|\PP_0|}.
  \end{equation}
  If $p\not\in\QQ\cup\VV$, then $\EE$ does not occur, and by Claim~\ref{clm:lb_p_not_visible}, $L\geq |\PP_0|-|\QQ|$.
  Hence, applying the union bound for (\ref{eq:prob_p_in_Q}) and (\ref{eq:prob_p_in_V})
  \begin{multline}\label{eq:L_p|p_in_W}
    \CondExp{L}{p\not\in\QQ\cup\VV}\cdot\Prob{p\not\in\QQ\cup\VV}
    \geq
    (|\PP_0|-|\QQ|)\cdot\paren{1-\frac{|\QQ|+|\VV|}{|\PP_0|}}
    \\ \geq
    \frac{(|\PP_0|-|\QQ|-|\VV|)^2}{|\PP_0|}.
  \end{multline}
  Now let
  \begin{displaymath}
    Z=\max\bigl\{|\VV|(|\VV|-1)/2,\,|\QQ|\cdot |\PP_0|,\,(|\PP_0|-|\QQ|-|\VV|)^2\bigr\}.
  \end{displaymath}
  Since either $|\VV|\geq |\PP_0|/3$ or $|\QQ|\geq |\PP_0|/3$ or $(|\PP_0|-|\QQ|-|\VV|)\geq |\PP_0|/3$, we have
  \begin{displaymath}
    Z\geq\min\left\{\frac{(|\PP_0|/3)(|\PP_0|/3-1)}{2},\ \frac{|\PP_0|^2}{3},\ (|\PP_0|/3)^2\right\}=\Omega(|\PP_0|^2).
  \end{displaymath}
  Summing up (\ref{eq:L_p|p_in_V}), (\ref{eq:L_p|p_in_R}), and (\ref{eq:L_p|p_in_W}), we obtain
  \begin{displaymath}
    \Exp{L}
    \geq
    \frac{1}{|\PP_0|}\cdot Z
    =
    \Omega(|\PP_0|).
  \end{displaymath}
\end{proof}

We are now ready to prove the main lemma.
Up to now, we assumed that the coin flips $\tc$ are fixed.
In the following we have to consider a random choice of $\tc\in\Omega^\infty$, and a random choice of $p\in\PP$.
We use the fact that the number of processes, $|\PP_i|$, that choose a counter object $i$ is highly concentrated around the expectation, $n/m=\sqrt{n}$.
Consequently it is unlikely that there is an index $i$ such that $|\PP_i|$ is much larger or smaller than $\sqrt{n}$.
\begin{claim}\label{clm:lb_expectations}
  Let $\mu=n/m$ and $\delta>0$ an arbitrarily small constant.
  For a randomly chosen $\tc$, with probability $1-o(1)$ it is true that
  \begin{displaymath}
    \forall 0\leq i<m:\ (1-\delta)\mu<|\PP_i|<(1+\delta)\mu.
  \end{displaymath}
\end{claim}
\begin{proof}
  First consider an arbitrary index $i\in\{0,\dots,m-1\}$, and let $Y_i=|\PP_i|$.
  Since each process $q\in\PP$ independently chooses an index $i_q\in\{0,\dots,m-1\}$, the distribution of the random variable $Y_i$ is identical to that of the binomial random variable $B(n,1/m)$.
  From Chernoff Bounds and $\mu=n/m=\sqrt{n}$ we obtain
  \begin{displaymath}
    \bProb{Y_i\leq (1-\delta)\mu\,\vee\,Y_i\geq (1+\delta)\mu}\leq
    2\cdot e^{-\Omega(\mu)}
    =
    e^{-\Omega(\sqrt{n})}.
  \end{displaymath}
  The claim now follows immediately from summing up this probability bound for $Y_0,\dots,Y_{m-1}$.
\end{proof}

\begin{proof}[Proof of Lemma~\ref{lem:weak-adversary-main}]
  Choose $\tc$ and $p$ uniformly at random and consider the history $H=H(\AA_p,p,\tc)$ obtained by a scheduling of the adversary $A_p$.
  Let $\Gamma$ denote the event that
  \begin{displaymath}
    (1-\delta)m<|P_0|,\dots,|\PP_{m-1}|<(1+\delta)m.
  \end{displaymath}
  Note that whether or not $\Gamma$ occurs depends only on $\tc$ and not on $p$.

  Let $L=L(\AA_p,p,\tc)$ denote the return value of $p$'s \FAI{} call.
  By Lemma~\ref{lem:L_vs_Pj} and since $\Gamma$ is independent of $p$,
  \begin{equation}\label{eq:L_given_Gamma}
    \CondExp{L}{\Gamma}=\Omega(m).
  \end{equation}

  By definition of $\AA_p$, if $p\in P_j$, then at any point there is at most one process active that is not in $\PP_{j}$.
  Hence, if $\Gamma$ occurs, then the maximum point contention, $K$, satisfies
  \begin{displaymath}
    K\leq\max_{0\leq i<m}|\PP_i|+1<(1+\delta)m+1,
  \end{displaymath}
  and since $K$ is an integer, $K\leq\lceil(1+\delta)m\rceil=K_{max}$.
  Hence, given that $\Gamma$ occurs, $X(\AA_p,p,\tc)=L$ and so applying Claim~\ref{clm:lb_expectations} and (\ref{eq:L_given_Gamma}),
  \begin{displaymath}
    \Exp{X(\AA_p,p,\tc)}
    \geq
    \CondExp{X(\AA_p,p,\tc)}{\Gamma}\cdot\Prob{\Gamma}
    =
    \CondExp{L}{\Gamma}\cdot\bparen{1-o(1)}
    =\Omega(m).
  \end{displaymath}
\end{proof}

%%%%% Done processing weakAdversary.tex

\paragraph{Acknowledgments}
We are grateful to the anonymous referees of STOC 2011 for their comments.
Thanks also to Faith Ellen for her detailed feedback.

\newpage
\bibliographystyle{alpha}

\begin{thebibliography}{GHHW07}

\bibitem[AAD{\etalchar{+}}93]{aadgms:snapshots}
Y.~Afek, H.~Attiya, D.~Dolev, E.~Gafni, M.~Merritt, and N.~Shavit.
\newblock Atomic snapsots of shared memory.
\newblock {\em J. ACM}, 40(4):873--890, 1993.

\bibitem[Abr88]{Abrahamson88_PODC}
Karl~R. Abrahamson.
\newblock On achieving consensus using a shared memory.
\newblock In {\em PODC '88: Proc. of the 6th annual ACM symposium on Principles
  of distributed computing}, pages 291--302, 1988.

\bibitem[AG88]{journals/ipl/AndersonG88}
James~H. Anderson and Mohamed~G. Gouda.
\newblock Atomic semantics of nonatomic programs.
\newblock {\em Inf. Process. Lett.}, 28(2):99--103, 1988.

\bibitem[Asp03]{Aspnes2003_DistrComp}
James Aspnes.
\newblock Randomized protocols for asynchronous consensus.
\newblock {\em Distributed Computing}, 16:165--175, 2003.

\bibitem[CIL87]{CIL1987_PODC}
Benny Chor, Amos Israeli, and Ming Li.
\newblock On processor coordination using asynchronous hardware.
\newblock In {\em PODC '87: Proc. of the 6th annual ACM symposium on Principles
  of distributed computing}, pages 86--97, 1987.

\bibitem[GHHW07]{ghhw:cas}
W.~Golab, V.~Hadzilacos, D.~Hendler, and P.~Woelfel.
\newblock Constant-rmr implementations of cas and other synchronization
  primitives using read and write operations.
\newblock In {\em PODC '07: Proc. of the 26th annual ACM symposium on
  Principles of distributed computing}, pages 0--0, 2007.

\bibitem[Gol10]{golab:phd}
W.~Golab.
\newblock {\em Constant-RMR Implementations of CAS and Other Synchronization
  Primitives Using Read and Write Operations}.
\newblock PhD thesis, University of Toronto, 2010.

\bibitem[Her91]{herl:wait}
M.~Herlihy.
\newblock Wait-free synchronization.
\newblock {\em ACM TOPLAS}, 13(1), January 1991.

\bibitem[HLM03]{her:of}
M.~Herlihy, V.~Luchangco, and M.~Moir.
\newblock Obstruction-free synchronization: Double-ended queues as an example.
\newblock In {\em ICDCS '03: Proc. of the 23rd International Conference on
  Distributed Computing Systems}, page 522, Washington, DC, USA, 2003. IEEE
  Computer Society.

\bibitem[HS08]{hs:art}
M.~Herlihy and N.~Shavit.
\newblock {\em The Art of Multiprocessor Programming}.
\newblock Morgan Kaufmann Publishers, first edition, 2008.

\bibitem[HW90]{her:lin}
M.~Herlihy and J.~M. Wing.
\newblock Linearizability: A correctness condition for concurrent objects.
\newblock {\em ACM TOPLAS}, 12(3):463--492, July 1990.

\bibitem[IL93]{il:timestamps}
A.~Israeli and M.~Li.
\newblock Bounded time-stamps.
\newblock {\em Distributed Computing}, 6(4):205--209, 1993.

\bibitem[Lyn96]{Lynch_DistributedAlgorithms1996}
Nancy Lynch.
\newblock {\em Distributed Algorithms}.
\newblock Morgan Kaufman, 1996.

\bibitem[VA86]{vit:atom}
P.~Vitanyi and B.~Awerbuch.
\newblock Atomic shared register access by asynchronous hardware.
\newblock In {\em Proc. of the 27th Annual Symposium on Foundations of Comupter
  Science}, pages 233--243. IEEE, 1986.

\bibitem[Vid88]{vid:registers}
K.~Vidyasankar.
\newblock Converting lamport's regular register to atomic register.
\newblock {\em Information Processing Letters}, 28(6):287--290, August 1988.

\bibitem[YA95]{yang:fast}
J.-H. Yang and J.~Anderson.
\newblock A fast, scalable mutual exclusion algorithm.
\newblock {\em Distributed Computing}, 9(1):51--60, August 1995.

\end{thebibliography}
%%%%% Processing file archiveMain.bbl
\newcommand{\etalchar}[1]{$^{#1}$}

%%%%% Done processing archiveMain.bbl

\newpage
%%%%% Processing file appendix.tex
\appendix

\let\appsection\section

\appsection{Additional Examples}
\label{sec-appendix-examples}

When, in a randomized distributed algorithm,
atomic operations are replaced by their linearizable implementations,
the scheduler's power can increase.
Even when the power of the scheduler of the implemented algorithm is constrained significantly from its
power in the atomic model, the scheduler can still create a probability distribution of computations that
is dramatically different from what is attainable in the atomic case.

The following are various additional examples of this phenomenon.

\paragraph{Implementing multi-reader/single-writer registers from single-reader/single-writer registers.}

Let $R$ be a two-reader atomic register initialized to 0,
accessed by writer $w$, and readers $r_1$ and $r_2$ that are executing the code:
\begin{quote}
$w$: $R$.\MRwrite(1); \quad  $cf$ := uniform-random$\{-1,1\}$; \quad   $R$.\MRwrite($cf$)  \\
$r_1$:  $R$.\MRread () \\
$r_2$:  $R$.\MRread ()
\end{quote}

Suppose the strong adversary is trying to minimize the value of $r_1$'s \MRread.
If $R$ is an atomic register,
then the adversary's best strategy is to have $r_1$ execute its
\MRread\ either before or after both of $w$'s \MRwrite\ operations.
In either case the expected value of $r_1$'s \MRread\ is 0.

Now suppose, instead, that $R$ is implemented using
the algorithm shown in Figure~\ref{fig:multireader_registers}.
This construction
uses 6 single-reader/single-writer registers:
$R_{w-r1}, R_{w-r2},R_{r1-r1},R_{r1-r2},R_{r2-r1},R_{r2-r2},$
where the subscript $x$-$y$ denotes a register written by $x$ and read by $y$.
Each is initialized to the initial value of $R$.

\begin{figure}
 \begin{minipage}[t]{.4\textwidth}
 \begin{function}[H]
   \caption{$R$.MRSWwrite($v$)}
    $seq := seq + 1$, $seq$ is a sequence number initialized to 1\;
    $R_{w-r1}.\xwrite(v,seq)$ \;
    $R_{w-r2}.\xwrite(v,seq)$

 \end{function}
 \end{minipage}\hfill
 \begin{minipage}[t]{.55\textwidth}
 \begin{function}[H]
   \caption{$R$.MRSWread()}
    code for process $r_i, i \in \{1,2\}$ \;
   $v[0], s[0] : = R_{w-ri}.\read $ \;
   $v[1], s[1] : = R_{r1-ri}.\read $ \;
   $v[2], s[2] : = R_{r2-Ri}.\read $ \;
   let $j$ be such that $s[j] = \max \{ s[0],s[1], s[2] \}$ \;
   $R_{ri-r1}.\xwrite(v[j],seq[j])$ \;
   $R_{ri-r2}.\xwrite(v[j],seq[j])$ \;
   return $v[j]$
 \end{function}
 \end{minipage}%
 \caption{Linearizable Implementation of MRSW Registers from SRSW Registers}%
 \label{fig:multireader_registers}
\end{figure}

Under this implementation, the adversary's could schedule as follows:
First one step of $r_1$: $R_{w-r1}$.\read, which will return 0;
then all of $w$'s steps (that is,  all of its first \MRwrite, its flip and its second \MRwrite).
If the value of $cf$ is 1,
then the adversary next schedules the rest of $r_1$'s steps,
followed by all $r_2$'s steps in its $R$.\MRread.
In this case the implementation of $R$.\MRread by $r_1$ returns 0.
If the value of $cf$ is -1,
then the adversary next schedules all $r_2$'s steps in its $R$.\MRread\
followed by the remainder of $r_1$'s $R$.\MRread.
In this case, $r_1$ will discover the updated value of $R$, when it
executes $R_{r2-r1}$.\read, and the implementation returns -1.
Hence the expected value returned by $r_1$ is -1/2.

Notice that this schedule is available even to the weak adversary.
It needs the coin flip value only after $w$ has completed all its operations.
So reducing the power of the adversary from strong to weak
does not curtail its power sufficiently to retain the expected behaviour of the algorithm
when $R$ is an atomic register.

This algorithm in Figure  \ref{fig:multireader_registers} is
the wait-free linearizable implementation
of multi-reader/multi-writer multivalued atomic registers from
single-reader/single-writer multivalued atomic registers
using unbounded sequence numbers
due to Vitanyi and Awerbuch (\cite{vit:atom})
when specialized to two readers and one writer.
The increased power of the weak adversary in the implementation over the power
of the strong adversary in the atomic case also extends to the case when there is more than one writer.

\paragraph{Implementing a queue with atomic increment objects.}
Let $Q$  be a queue object, initially empty and accessed by  $q_1$, $q_2$ and $p$ that are executing the code:

\begin{quote}
$q_0$:   $Q$.\Enqueue(0)\\
$q_1$:   $Q$.\Enqueue(1)\\
$p$: $Q$.\Enqueue(2);\ $cf$ := uniform-random$\{0,1\}$; \ $Q$.\Dequeue; $Q$.\Dequeue; $Q$.\Dequeue\\
\end{quote}

The adversary's goal is to achieve the following:
\begin{enumerate}
 \item[(a)] all of $p$'s \Dequeue\ operations succeed,
 \item[(b)] the \Dequeue operation that returns 1 precedes the \Dequeue operation that returns 2, and
 \item[(c)] the return value of $p$'s first \Dequeue\ operation equals the result of the flip.
\end{enumerate}

Even if the adversary is strong, in order to achieve (b), it must schedule $q_1$'s \Enqueue operation before $p$'s flip operation.
Hence, by the time the flip occurs, the decision whether 0 or 1 is in front of the queue has been made, and cannot be changed by the adversary.
Therefore, the probability that $p$'s first \Dequeue operation returns the value of the flip is at most $1/2$.

\begin{figure}[bt]
 \begin{center}
 \begin{minipage}[t]{.4\textwidth}
 \begin{function}[H]
   \caption{$Q$.enq($v$)}
    $pos :=$ $tail$.\fetchInc()  \;
    $item[pos]$.\xwrite($v$)
 \end{function}
 \end{minipage}\hfill%
 \begin{minipage}[t]{.4\textwidth}
   \begin{function}[H]
   \caption{$Q$.deq($v$)}
   \While{\True}{
     $max:=tail$.\read{}\;
     \For{$i=0\dots max-1$}{
       $v:=temp[i].\fetchSet$\;
       \IlIf{$i\neq\bot$}{\Return{$v$}}
      }
    }
   \end{function}
 \end{minipage}
 \end{center}
 \caption{Implementation of the Herlihy-Wing Queue}%
 \label{fig:queue}
\end{figure}
Herlihy and Wing give a linearizable implementation of
a queue using read-modify-write base objects
that support the operations \fetchSet, \fetchInc and \read
\cite{her:lin}.
The queue is represented by an unbounded array of items with a tail pointer.
The implementation is shown in Figure \ref{fig:queue}.

If the above algorithm uses this queue implementation, then the weak adversary could schedule as shown in Figure~\ref{fig:queue-schedule}.
Here, the left bar in the drawing of a queue method call denotes that method call's first shared memory access, and the right bar its last.
I.e., the \fetchInc operations of the \Enqueue operations occur in the order $q_0$, $q_1$, $p$.
And both, $q_1$ and $p$ execute their $\xwrite$ before the flip happens.
If the result of the flip is 0, then $q_0$ writes immediately, before any \Dequeue operation starts.
Therefore, the first \Dequeue operation will return 0.
If the result of the flip is 1, then $q_0$'s \xwrite is delayed until after the first \Dequeue operation completed.
In this case the first \Dequeue operation will return 1.
In either case, the second \Dequeue operation will return 0 or 1, and the third will return 2.
Hence, even the weak adversary can achieve with probability 1 that (a), (b), and (c) are satisfied.
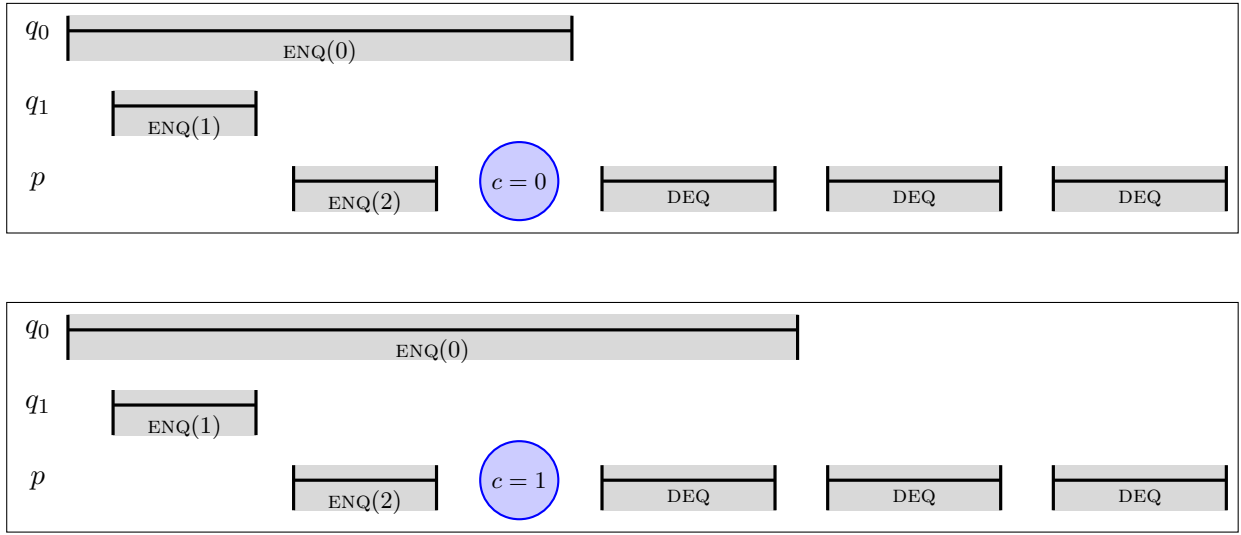
\begin{figure}[bt]
%%%%% Processing file queue_example.tex
\begin{tikzpicture}[
>=latex, %use latex arrows
very thick,
framed,
on grid,
auto,
]
\footnotesize

\newcommand{\outerop}[4][]{%
  \fill[fill=gray!30] ($(#2)+(0,-.4)$) rectangle ($(#3)+(0,.2)$);
  \draw[-]  (#2) -- node[swap,#1] {#4} (#3);
  \draw[very thick] ($(#2)+(0,-.4)$) -- ($(#2)+(0,.2)$);
  \draw[very thick] ($(#3)+(0,-.4)$) -- ($(#3)+(0,.2)$);
}

\newcommand{\innerop}[4][]{%
  \fill[fill=gray!70] ($(#2)+(0,-.4)$) rectangle ($(#3)+(0,.4)$);
  \draw[<->]  (#2) -- node[#1] {#4} (#3);
}

\node (r) {\normalsize $q_1$};
\node[below = 1 of r] (q) {\normalsize $p$};
\node[above = 1 of r] (p) {\normalsize $q_{0}$};

\coordinate[right = .4 of p ] (pScanBegin);
\coordinate[right = 7.1 of p ] (pScanEnd) ;
\outerop{pScanBegin}{pScanEnd}{$\Enqueue(0)$}

\coordinate[right = 3.4 of q ] (qUpdate2Begin) ;
\coordinate[right = 1.9 of qUpdate2Begin ] (qUpdate2End) ;
\outerop{qUpdate2Begin}{qUpdate2End}{$\Enqueue(2)$}

\node[circle,text centered,fill=blue!20,draw=blue!200,thick,
  right = 6.4 of q] (coinflip) {$c=0$};

\coordinate[right = 7.5 of q] (qDeq1Begin);
\coordinate[right = 2.3 of qDeq1Begin ] (qDeq1End);
\outerop{qDeq1Begin}{qDeq1End}{$\Dequeue$}

\coordinate[right = 10.5 of q] (qDeq2Begin);
\coordinate[right = 2.3 of qDeq2Begin ] (qDeq2End);
\outerop{qDeq2Begin}{qDeq2End}{$\Dequeue$}

\coordinate[right = 13.5 of q] (qDeq3Begin);
\coordinate[right = 2.3 of qDeq3Begin ] (qDeq3End);
\outerop{qDeq3Begin}{qDeq3End}{$\Dequeue$}

\coordinate[right = 1 of r ] (rUpdate1Begin) ;
\coordinate[right = 1.9 of rUpdate1Begin ] (rUpdate1End) ;
\outerop{rUpdate1Begin}{rUpdate1End}{$\Enqueue(1)$}
\end{tikzpicture}

\vskip2\bigskipamount
\begin{tikzpicture}[
>=latex, %use latex arrows
very thick,
framed,
on grid,
auto,
]
\footnotesize

\newcommand{\outerop}[4][]{%
  \fill[fill=gray!30] ($(#2)+(0,-.4)$) rectangle ($(#3)+(0,.2)$);
  \draw[-]  (#2) -- node[swap,#1] {#4} (#3);
  \draw[very thick] ($(#2)+(0,-.4)$) -- ($(#2)+(0,.2)$);
  \draw[very thick] ($(#3)+(0,-.4)$) -- ($(#3)+(0,.2)$);
}

\newcommand{\innerop}[4][]{%
  \fill[fill=gray!70] ($(#2)+(0,-.4)$) rectangle ($(#3)+(0,.4)$);
  \draw[<->]  (#2) -- node[#1] {#4} (#3);
}

\node (r) {\normalsize $q_1$};
\node[below = 1 of r] (q) {\normalsize $p$};
\node[above = 1 of r] (p) {\normalsize $q_{0}$};

\coordinate[right = .4 of p ] (pScanBegin);
\coordinate[right = 10.1 of p ] (pScanEnd) ;
\outerop{pScanBegin}{pScanEnd}{$\Enqueue(0)$}

\coordinate[right = 3.4 of q ] (qUpdate2Begin) ;
\coordinate[right = 1.9 of qUpdate2Begin ] (qUpdate2End) ;
\outerop{qUpdate2Begin}{qUpdate2End}{$\Enqueue(2)$}

\node[circle,text centered,fill=blue!20,draw=blue!200,thick,
  right = 6.4 of q] (coinflip) {$c=1$};

\coordinate[right = 7.5 of q] (qDeq1Begin);
\coordinate[right = 2.3 of qDeq1Begin ] (qDeq1End);
\outerop{qDeq1Begin}{qDeq1End}{$\Dequeue$}

\coordinate[right = 10.5 of q] (qDeq2Begin);
\coordinate[right = 2.3 of qDeq2Begin ] (qDeq2End);
\outerop{qDeq2Begin}{qDeq2End}{$\Dequeue$}

\coordinate[right = 13.5 of q] (qDeq3Begin);
\coordinate[right = 2.3 of qDeq3Begin ] (qDeq3End);
\outerop{qDeq3Begin}{qDeq3End}{$\Dequeue$}

\coordinate[right = 1 of r ] (rUpdate1Begin) ;
\coordinate[right = 1.9 of rUpdate1Begin ] (rUpdate1End) ;
\outerop{rUpdate1Begin}{rUpdate1End}{$\Enqueue(1)$}
\end{tikzpicture}
%%%%% Done processing queue_example.tex
\caption{A ``bad'' scheduling for the queue algorithm.}
\label{fig:queue-schedule}
\end{figure}

%%%%% Done processing appendix.tex

\end{document}